\documentclass{amsart}

\usepackage{amsmath, amssymb, amscd}

\usepackage{enumerate}
\usepackage[pdftex]{hyperref} 
\hypersetup{pdfstartview=FitH,
	pdfpagemode=UseNone}

\newcommand{\field}[1]{\mathbb{#1}}
\newcommand{\N}{\field{N}}
\newcommand{\Z}{\field{Z}}
\newcommand{\Q}{\field{Q}}

\newcommand{\VV}{\field{V}}
\newcommand{\F}{\field{F}}

\newcommand{\term}[1]{\boldsymbol{#1}}

\newcommand{\Om}{\Omega}
\newcommand{\pder}[2]{\partial_{#2}#1}
\newcommand{\cJ}{\mathcal{J}}
\newcommand{\cH}{\mathcal{H}}

\DeclareMathOperator{\ch}{char}
\DeclareMathOperator{\rk}{rk}
\DeclareMathOperator{\poly}{poly}
\DeclareMathOperator{\trdeg}{trdeg}
\DeclareMathOperator{\size}{size}

\DeclareMathOperator{\sparse}{sp} 
\DeclareMathOperator{\WJ}{WJ} 
\DeclareMathOperator{\WJP}{WJP} 
\DeclareMathOperator{\W}{W} 
\DeclareMathOperator{\J}{J} 
\DeclareMathOperator{\E}{E}
\DeclareMathOperator{\Fil}{Fil}
\DeclareMathOperator{\V}{V} 
\DeclareMathOperator{\Restr}{R} 
\DeclareMathOperator{\Frob}{F} 

\DeclareMathOperator{\sep}{sep} 
\DeclareMathOperator{\insep}{insep} 

\DeclareMathOperator{\NP}{NP}
\DeclareMathOperator{\FP}{FP}
\DeclareMathOperator{\PH}{PH}
\DeclareMathOperator{\sP}{\# P}
\DeclareMathOperator{\coeff}{coeff}
\DeclareMathOperator{\ZPP}{ZPP}
\DeclareMathOperator{\CeqP}{C_=P}

\providecommand{\abs}[1]{\left\lvert#1\right\rvert}
\providecommand{\ol}[1]{\overline{#1}}

\makeatletter
\newtheorem*{rep@theorem}{\rep@title}
\newcommand{\newreptheorem}[2]{%
\newenvironment{rep#1}[1]{%
 \def\rep@title{#2 \ref{##1}}%
 \begin{rep@theorem}}%
 {\end{rep@theorem}}}
\makeatother

\theoremstyle{plain}
\newtheorem{thm}{Theorem}

\newtheorem{cor}[thm]{Corollary}
\newtheorem{lem}[thm]{Lemma}

\newreptheorem{thm}{Theorem} 
\newreptheorem{lem}{Lemma}

\theoremstyle{definition}
\newtheorem{defn}[thm]{Definition}

\theoremstyle{remark}
\newtheorem{rem}[thm]{Remark}

\newcommand{\comment}[1]{}

\comment{
Algebraic Independence in Positive Characteristic -- A p-Adic Calculus

Johannes Mittmann, Nitin Saxena, Peter Scheiblechner

Hausdorff Center for Mathematics, Endenicher Allee 62, 53115 Bonn, Germany

johannes.mittmann@hcm.uni-bonn.de
nitin.saxena@hcm.uni-bonn.de
peter.scheiblechner@hcm.uni-bonn.de

A set of multivariate polynomials, over a field of zero or large characteristic, can be tested for algebraic independence by the well-known Jacobian criterion. For fields of other characteristic p>0, there is no analogous characterization known. In this paper we give the first such criterion. Essentially, it boils down to a non-degeneracy condition on a lift of the Jacobian polynomial over (an unramified extension of) the ring of p-adic integers.

Our proof builds on the de Rham-Witt complex, which was invented by Illusie (1979) for crystalline cohomology computations, and we deduce a natural generalization of the Jacobian. This new avatar we call the Witt-Jacobian. In essence, we show how to faithfully differentiate polynomials over F_p (i.e. somehow avoid dx^p/dx=0) and thus capture algebraic independence.

We apply the new criterion to put the problem of testing algebraic independence in the complexity class NP^#P (previously best was PSPACE). Also, we give a modest application to the problem of identity testing in algebraic complexity theory.

}

\begin{document}

\title[Algebraic Independence in Positive Characteristic]{Algebraic Independence in Positive Characteristic -- A $p$-Adic Calculus}

\author[J. Mittmann]{Johannes Mittmann}
\address{Hausdorff Center for Mathematics\\ Endenicher Allee 62\\ D-53115 Bonn\\ Germany}
\email{johannes.mittmann@hcm.uni-bonn.de}
\author[N. Saxena]{Nitin Saxena}
\address{Hausdorff Center for Mathematics\\ Endenicher Allee 62\\ D-53115 Bonn\\ Germany}
\email{nitin.saxena@hcm.uni-bonn.de}
\author[P. Scheiblechner]{Peter Scheiblechner}
\address{Hausdorff Center for Mathematics\\ Endenicher Allee 62\\ D-53115 Bonn\\ Germany}
\email{peter.scheiblechner@hcm.uni-bonn.de}



\begin{abstract}
A set of multivariate polynomials, over a field of zero or large characteristic, 
can be tested for algebraic independence by the well-known Jacobian criterion.
For fields of other characteristic $p>0$, there is no analogous characterization
known. In this paper we give the first such criterion. 
Essentially, it boils down to a non-degeneracy condition on a lift of the Jacobian
polynomial over (an unramified extension of) the ring of $p$-adic integers.

Our proof builds on the de Rham-Witt complex, which was invented by Illusie (1979) for crystalline cohomology 
computations, and we deduce a natural generalization of the Jacobian. This new avatar we call the Witt-Jacobian. 
In essence, we show how to faithfully differentiate polynomials over $\F_p$ (i.e.~somehow avoid 
$\partial x^p/\partial x=0$) and thus capture algebraic independence.

We apply the new criterion to put the problem of testing algebraic independence in the complexity class
$\NP^{\sP}$ (previously best
was PSPACE). Also, we give a modest application to the problem of identity testing in 
algebraic complexity theory.
\end{abstract}

\maketitle

\section{Introduction}


Polynomials $\term{f}=\{f_1,\ldots,f_m\}\subset k[x_1,\ldots,x_n]$ are called {\em algebraically independent} 
over a field $k$, if there is no
nonzero $F\in k[y_1,\ldots,y_m]$ such that $F(\term{f})=0$. Otherwise, they are algebraically 
{\em dependent} and $F$ is an {\em annihilating} polynomial. Algebraic independence is a fundamental concept
in commutative algebra. It is to polynomial rings what {\em linear} independence is to vector spaces. Our paper is
motivated by the computational aspects of this concept. 

A priori it is not clear whether, for given {\em explicit} polynomials, one can test algebraic independence {\em effectively}.
But this is possible -- by Gr\"{o}bner bases, or, by invoking Perron's degree bound on the annihilating polynomial 
\cite{bib:Per27} and 
finding a possible $F$. Now, can this be done {\em efficiently} (i.e.~in polynomial time)? It can be seen that both the 
above algorithmic techniques take {\em exponential} time, though the latter gives a PSPACE algorithm.
Hence, a different approach is needed for a faster algorithm, and here enters Jacobi \cite{bib:Jac41}.
The \emph{Jacobian} of the polynomials~$\term{f}$ is the matrix 
$\cJ_{\term{x}}(\term{f}) := (\pder{f_i}{x_j})_{m \times n}$, 
where $\pder{f_i}{x_j} = \partial f_i/ \partial x_j$
is the partial derivative of $f_i$ with respect to~$x_j$.
It is easy to see that for $m>n$ the $\term{f}$ are dependent, so
 we always assume $m\le n$. 
Now, the Jacobian criterion says: The matrix is of {\em full} rank over the function field iff $\term{f}$ 
are algebraically independent (assuming zero or large characteristic, see \cite{bib:BMS11}). Since the rank of this matrix
can be computed by its {\em randomized} evaluations \cite{bib:Sch80,bib:DGW09},
we immediately get
a randomized polynomial time algorithm. The only question left is -- What about the `other' prime characteristic fields? In those 
situations nothing like the Jacobian criterion was known. Here we propose the first such criterion that works for {\em all} 
prime characteristic. In this sense we make partial progress on the algebraic independence question for `small fields' 
\cite{bib:DGW09}, but we do not yet know how to check this criterion in polynomial time. We do, however, improve 
the complexity of algebraic independence testing from PSPACE to $\NP^{\sP}$.

The $m\times m$ minors of the Jacobian we call {\em Jacobian polynomials}. So the criterion can be rephrased: 
One of the Jacobian polynomials is nonzero iff $\term{f}$ are algebraically independent (assuming zero or large characteristic).
We believe that finding a Jacobian-type polynomial that captures algebraic independence in any characteristic $p>0$ is a 
natural question in algebra and geometry. Furthermore, Jacobian has recently found several applications in complexity theory --
circuit lower bound proofs \cite{bib:Kal85,bib:ASSS11}, 
pseudo-random objects construction \cite{bib:DGW09,bib:Dvi09}, 
identity testing \cite{bib:BMS11,bib:ASSS11}, 
cryptography \cite{bib:DGRV11},
program invariants \cite{bib:Lvo84,bib:Kay09},
and control theory \cite{bib:For91,bib:DF92}.
Thus, a suitably effective Jacobian-type criterion is desirable to make these applications work 
for any field. The criterion presented here is not yet effective enough, nevertheless, it is able to solve a modest case of identity 
testing that was left open in \cite{bib:BMS11}.

In this paper, the new avatar of the Jacobian polynomial is called a {\em Witt-Jacobian}. For polynomials 
$\term{f}=\{f_1,\ldots,f_n\}\subset \F_p[x_1,\ldots,x_n]$ we simply lift the coefficients of $\term{f}$ to the {\em $p$-adic 
integers} $\hat{\Z}_p$, to get the {\em lifted} polynomials $\term{\hat{f}}\subset \hat{\Z}_p[x_1,\ldots,x_n]$. Now, for 
$\ell\ge1$, the $\ell$-th Witt-Jacobian polynomial is $\WJP_\ell:=
(\hat{f}_1 \cdots \hat{f}_n)^{p^{\ell-1}-1} (x_1 \cdots x_n)\cdot \det \cJ_{\term{x}}(\term{\hat{f}})$. Hence, the 
Witt-Jacobian is just a suitably `scaled-up' version of the Jacobian polynomial
over the integral domain $\hat{\Z}_p$.
E.g., if $n=1, f_1=x_1^p$, then $\WJP_\ell=(x_1^p)^{p^{\ell-1}-1} (x_1)\cdot (px_1^{p-1})=
px_1^{p^\ell}$ which is a nonzero $p$-adic polynomial. Thus, Witt-Jacobian avoids mapping $x_1^p$ to zero. However, the flip side
is that a lift of the polynomial $f_1=0$, say, $\hat{f}_1=px_1^p$ gets mapped to $\WJP_\ell=(px_1^p)^{p^{\ell-1}-1} (x_1)\cdot 
(p^2x_1^{p-1})=p^{(p^{\ell-1}+1)}x_1^{p^\ell}$ which is also a nonzero $p$-adic polynomial. This shows that a Witt-Jacobian criterion 
cannot simply hinge on the 
zeroness of $\WJP_\ell$ but has to be much more subtle. Indeed, we show that the
terms in $\WJP_\ell$ carry {\em precise} information about the algebraic independence of $\term{f}$. In particular, in the two examples above, our
Witt-Jacobian criterion checks whether the coefficient of the monomial $x_1^{p^\ell}$ in $\WJP_\ell$ is divisible by~$p^{\ell}$ (which is true in the
second example, but not in the first for $\ell\ge2$). It is the magic of abstract {\em differentials} that such a weird explicit property could be 
formulated at all, let alone proved.

\subsection{Main results}

We need some notation to properly state the results. Denote $\Z^{\ge0}$ by $\N$.
Let $[n] := \{1,\ldots,n\}$, and the set of all $r$-subsets of $[n]$ be denoted by 
$\tbinom{[n]}{r}$. 
If $I \in \tbinom{[n]}{r}$, the bold-notation $\term{a}_I$ will be a short-hand
for $a_i$, $i\in I$, and we write $\term{a}_i$ for $\term{a}_{[i]}$.
Let $k/\F_p$ be an algebraic field extension, and $\W(k)$ be the ring of
Witt vectors of~$k$ ($\W(k)$ is just a `nice' extension of $\hat{\Z}_p$).
Define the $\F_p$-algebra $A:=k[\term{x}_n]$ and the $p$-adic-algebra $B:=\W(k)[\term{x}_n]$. 
For a nonzero $\alpha\in\N^n$ denote by $v_p(\alpha)$ the maximal $v\in\N$ with $p^v|\alpha_i, i\in[n]$.
Set $v_p(\term{0}):=\infty$. 

{\bf [Degeneracy]}
    We call $f \in B$ {\em degenerate} if the coefficient of $\term{x}^\alpha$ in $f$ is divisible by $p^{v_p(\alpha)+1}$
	for all $\alpha \in \N^n$. For $\ell\in\N$, $f$ is called \emph{$(\ell+1)$-degenerate} if
	the coefficient of $\term{x}^\alpha$ in $f$ is divisible by $p^{\min\{v_p(\alpha), \ell\}+1}$
	for all $\alpha \in \N^n$.

We could show for polynomials $\term{f}_r\in A$ and their $p$-adic lifts $\term{g}_r\in B$, that
if~$\term{f}_r$  are algebraically dependent,
then for any $r$ variables $\term{x}_I$, $I \in \tbinom{[n]}{r}$,  the $p$-adic polynomial 
$\bigl( {\textstyle\prod_{j \in I} x_j} \bigr) \cdot \det\cJ_{\term{x}_I}(\term{g}_r)$ is degenerate. This
would have been a rather elegant criterion,
if the converse did not fail
(see Theorem \ref{thm:Necessity}). Thus, we need to look at a more complicated
polynomial (and use the graded version of degeneracy).

{\bf [Witt-Jacobian polynomial]}
Let $\ell\in\N$, $\term{g}_r \in B$, and $I \in \tbinom{[n]}{r}$. 
	We call
	\[
		\WJP_{\ell+1, I}(\term{g}_r) := (g_1 \dotsb g_r)^{p^{\ell}-1} 
		\bigl( {\textstyle\prod_{j \in I} x_j} \bigr)
		\cdot \det \cJ_{\term{x}_I}(\term{g}_r) \in B
	\]
	the \emph{$(\ell+1)$-th Witt-Jacobian polynomial of $\term{g}_r$ w.r.t.~$I$}.

\begin{thm}[Witt-Jacobian criterion] \label{thm:ExplicitWJCriterion}
	Let $\term{f}_r \in A$ be of degree at most $\delta \ge 1$, and fix
	$\ell \ge \lfloor r \log_p \delta \rfloor$. 
	Choose $\term{g}_r \in B$ such that $\forall i\in[r], f_i \equiv g_i \pmod{pB}$. 
	
	Then, $\term{f}_r$ are algebraically independent over $k$ if and only if there exists 
	$I \in \tbinom{[n]}{r}$ such that $\WJP_{\ell+1, I}(\term{g}_r)$ is not	$(\ell+1)$-degenerate.
\end{thm}

If $p>\delta^r$, this theorem subsumes the Jacobian criterion (choose $\ell=0$).
In computational situations we are given $\term{f}_r\in A$, say, explicitly. Of course, we can efficiently
lift them to $\term{g}_r\in B$. But $\WJP_{\ell+1, I}(\term{g}_r)$ may have 
{\em exponential} sparsity (number of nonzero monomials), even for $\ell=1$.
This makes it difficult to test the Witt-Jacobian polynomial efficiently for $2$-degeneracy. 
While we improve the basic upper bound of PSPACE for this problem,
there is some evidence that the {\em general} $2$-degeneracy 
problem is outside the polynomial hierarchy \cite{bib:M12} (Theorem \ref{thm:CPHard}).

\begin{thm}[Upper bound] \label{thm:UpperBound}
Given arithmetic circuits $\term{C}_r$ computing in $A$, the problem of testing algebraic independence of 
polynomials $\term{C}_r$ is in the class $\NP^{\sP}$.      
\end{thm}

We are in a better shape when $\WJP_{\ell+1, I}(\term{g}_r)$ is relatively sparse,
which happens, for instance, when $\term{f}_r$ have `sub-logarithmic' sparsity.
This case can be applied to the question
of {\em blackbox identity testing}\,: We are given an arithmetic circuit $C\in \F_p[\term{x}_n]$ via a blackbox,
and we need to decide whether $C=0$. Blackbox access means that we can only
evaluate $C$ over field extensions of $\F_p$. Hence, blackbox identity testing
boils down to efficiently constructing a {\em hitting-set} $\cH\subset\ol{\F}_p^n$
such that any nonzero $C$ (in our circuit family) has an $\term{a}\in\cH$ with
$C(\term{a})\ne0$. Designing efficient hitting-sets is an outstanding open
problem in complexity theory, see \cite{bib:SS95,bib:Sax09,bib:SY10,bib:ASSS11}
and the references therein. We apply the Witt-Jacobian criterion to 
the following case of identity testing.

\begin{thm}[Hitting-set] \label{thm:HittingSet}
Let $\term{f}_m \in A$ be $s$-sparse polynomials of degree $\le\delta$, transcendence degree $\le r$, 
and assume $s,\delta,r\ge 1$. Let $C \in k[\term{y}_m]$ such that the degree of $C(\term{f}_m)$ is bounded by 
$d$. We can construct a (hitting-)set $\mathcal{H}\subset\ol{\F}_p^n$ in $\poly\bigl( (nd)^r, 
(\delta rs)^{r^2 s}\bigr)$-time such that: If $C(\term{f}_m)\ne0$ then 
$\exists\boldsymbol{a} \in \mathcal{H}, \bigl(C(\term{f}_m)\bigr)(\boldsymbol{a}) \neq 0$. 
\end{thm}

An interesting parameter setting is $r=O(1)$ and $s=O(\log d/r^2\log(\delta r \log d))$. In other
words, we have an efficient hitting-set, when $\term{f}_m$ have constant transcendence degree and 
sub-logarithmic sparsity. This is new, though, for zero and large characteristic, a much better
result is in \cite{bib:BMS11} (thanks to the classical Jacobian).

\subsection{Our approach}

Here we sketch the ideas for proving Theorem \ref{thm:ExplicitWJCriterion}, without going 
into the definitions and technicalities (those come later in plenty). The central tool in the 
proof is the {\em de Rham-Witt complex} which was invented by Illusie, for $\F_p$-ringed topoi, in the seminal work 
\cite{bib:Il79}.  
While it is fundamental for several cohomology theories for schemes in characteristic $p>0$
(see the beautiful survey \cite{bib:Il94}), we focus here on its algebraic strengths only.
We will see that it is just the right machinery, though quite heavy, to 
churn a criterion. We lift a polynomial $f\in A$ to a more `geometric' ring $\W(A)$, via the
{\em Teichm\"uller lift} $[f]$. This process is the same functor that builds $\hat{\Z}_p$ from $\F_p$ 
\cite{bib:Ser79}. The formalization of differentiation in this ring is by the
$\W(A)$-module of {\em K\"ahler differentials} $\Om_{\W(A)}^1$~\cite{bib:Eis95}.
Together with its {\em exterior powers} it provides a fully-fledged linear
algebra structure, the {\em de Rham complex} $\Om_{\W(A)}^{\cdot}$. But this is
all in zero characteristic and we have to do more to correctly extract the
properties of $A$ -- which has characteristic $p$.

The ring $\W(A)$ admits a natural {\em filtration} by ideals
$\V^\ell\W(A)\supseteq p^\ell\W(A)$, so we
have {\em length-$\ell$ Witt vectors} $\W_\ell(A):= \W(A)/\V^\ell\W(A)$.
This filtration is inherited by $\Om_{\W(A)}^{\cdot}$,
and a suitable quotient defines the {\em de Rham-Witt complex}
$\W_\ell\Om_A^\cdot$ of
 $\W_\ell(A)$-modules, and the {\em de Rham-Witt pro-complex} 
$\W_\bullet\Om_A^\cdot$. This is still an abstractly defined object, but it can be explicitly realized as a 
subspace of the algebra $B':=\cup_{i\ge0}\W(k)[\term{x}_n^{p^{-i}}]$ (a {\em perfection} of $B$). Illusie
defined a subalgebra $\E^0\subset B'$ that is `almost' isomorphic to $\W(A)$,
and could then identify a {\em differential graded algebra} $\E\subset\Om_{B'}^\cdot$
such that a suitable quotient $\E_\ell:=\E/\Fil^\ell \E$ realizes $\W_\ell\Om_A^\cdot$.

To prove Theorem \ref{thm:ExplicitWJCriterion} we consider the {\em Witt-Jacobian differential}
$\WJ_\ell:= d [f_1] \wedge \dotsb \wedge d [f_r] \in \W_\ell\Om^r_A$. By studying the behavior of 
$\W_\ell\Om^r_A$ as we move from $A$ to an extension ring, we show that $\WJ_\ell$ vanishes 
iff $\term{f}_r$ are algebraically dependent. The concept of {\em \'etale} extension is really useful here 
\cite{bib:Milne80}. In our situation, it corresponds to a {\em separable} field extension. We try to `force' 
separability, and here Perron-like Theorem \ref{thm:PerronsTheorem}  helps to bound $\ell$.
Next, we realize $\WJ_\ell$
as an element of $\E_\ell^r$. It is here where the
{\em Witt-Jacobian polynomials} $\WJP_{\ell,I}$ appear and satisfy: $\WJ_\ell=0$ iff its explicit version is in
$\Fil^\ell \E^r$ iff $\WJP_{\ell,I}$ is $\ell$-degenerate for all $I$.

The idea in Theorem \ref{thm:UpperBound} is that,
by the Witt-Jacobian criterion, the given polynomials are algebraically independent
iff some $\WJP_{\ell+1,I}$ has some monomial~$\term{x}^\alpha$ whose coefficient is
{\em not} divisible by $p^{\min\{v_p(\alpha), \ell\}+1}$. An NP machine can `guess'~$I$ and 
$\alpha$, while computing the coefficient is harder. We do the latter following
an idea of~\cite{bib:KS11} by evaluating the exponentially large sum in an
interpolation formula using a $\sP$-oracle.
In this part the isomorphism between
 $\W_{\ell+1}(\F_{p^t})$ 
and the handier {\em Galois ring} $G_{\ell+1,t}$ \cite{bib:Rag69, bib:Wan03} allows to evaluate $\WJP_{\ell+1,I}$.

The main idea in Theorem \ref{thm:HittingSet} is that non-$\ell$-degeneracy of
$\WJP_{\ell,I}$ is preserved under evaluation of the variables $\term{x}_{[n]\setminus I}$.
This implies with \cite{bib:BMS11} that algebraically independent $\term{f}_r$ can be made
$r$-variate efficiently without affecting the 
zeroness of $C(\term{f}_r)$.
The existence of the claimed hitting-sets follows easily from \cite{bib:Sch80}.

\subsection{Organization}
In \S \ref{sec:Preliminaries} we introduce all necessary preliminaries about algebraic independence 
and transcendence degree
(\S \ref{sec:AlgebraicIndependenceTrdeg}), derivations, differentials and the de Rham complex
(\S \ref{sec:DerivationsDifferentialsDeRhamComplex}), separability (\S \ref{sec:Separability}), 
the ring of Witt vectors (\S \ref{sec:WittVectors}) and the de Rham-Witt complex (\S 
\ref{sec:DeRhamWittComplex} and \S \ref{sec:DeRhamWittComplexPolynomialRing}). To warm up the
concept of differentials we discuss the classical Jacobian criterion in a `modern' language in \S 
\ref{sec:SeparabilityClassicalJacobianCriterion}.

Our main results are contained in \S\ref{sec:WJCriterion}.
In \S \ref{sec:WJDifferential} we define the Witt-Jacobian differential
and prove the abstract Witt-Jacobian criterion, and in \S \ref{sec:WJPolynomial}
we derive its explicit version Theorem \ref{thm:ExplicitWJCriterion}.
In \S \ref{sec:UpperBound} and \S \ref{sec:ApplicationPIT} we prove
Theorems \ref{thm:UpperBound} resp.~\ref{thm:HittingSet}.
To save space we have skipped several worthy references and moved some
proofs to Appendix \ref{app:Proofs}.

\section{Preliminaries} \label{sec:Preliminaries}

Unless stated otherwise, a ring in this paper is commutative with unity.
For integers $r \le n$, we write $[r,n] := \{r, r+1, \dotsc, n\}$. 

\subsection{Algebraic independence and transcendence degree} \label{sec:AlgebraicIndependenceTrdeg}

Let $k$ be a field and let $A$ be a $k$-algebra. Elements $\term{a}_r \in A$ are called
\emph{algebraically independent over $k$} if $F(\term{a}_r) \neq 0$ for all
nonzero polynomials $F \in k[\term{y}_r]$. For a subset $S \subseteq A$,
the \emph{transcendence degree of $S$ over $k$} is defined as
$\trdeg_k(S) := \sup\{ \#T \,\vert\; T \subseteq S\text{ finite and algebraically independent over }k\}$.
For an integral domain $A$ we have $\trdeg_k(A)=\trdeg_k(Q(A))$,
where $Q(A)$ denotes the quotient field of $A$.

Now let $k[\term{x}] = k[\term{x}_n]$ be a polynomial ring over $k$. We have the following 
effective criterion for testing algebraic independence, which is stronger than the classical
Perron's bound \cite{bib:Per27}. We prove it in \S\ref{app:Preliminaries} using 
\cite[Corollary 1.8]{bib:K96}.

\begin{thm}[Degree bound] \label{thm:PerronsTheorem}
Let $k$ be a field, $\term{f}_{n} \in k[\term{x}]$ be algebraically independent, and set 
$\delta_i := \deg(f_i)$ for $i \in [n]$. 
Then $[k(\term{x}_n):k(\term{f}_n)] \le \delta_1 \dotsb \delta_{n}$.
\end{thm}

%


\subsection{Differentials and the de Rham complex} \label{sec:DerivationsDifferentialsDeRhamComplex}

Let $R$ be a ring and let $A$ be an $R$-algebra. The \emph{module of K\"ahler differentials}
of $A$ over $R$, denoted by $\Om_{A/R}^1$, is the $A$-module generated by the set
of symbols $\{da \,\vert\; a \in A\}$ subject to the relations
\[
	d(ra+sb) = r\,da + s\,db \;\; (\text{$R$-linearity}), \quad
	d(ab) = a\,db + b\,da \;\; (\text{Leibniz rule})
\]
for all $r,s \in R$ and $a,b \in A$. The map $d \colon A \rightarrow \Om_{A/R}^1$
defined by $a \mapsto da$ is an $R$-derivation called the \emph{universal $R$-derivation}
of $A$.

For $r \ge 0$, let $\Om_{A/R}^r := \bigwedge^r \Om_{A/R}^1$ be the $r$-th exterior
power over $A$.
The universal derivation $d \colon A=\Om_{A/R}^0 \rightarrow \Om_{A/R}^1$ extends
to the {\em exterior derivative} $d^r \colon \Om_{A/R}^r \to \Om_{A/R}^{r+1}$ by
$d^r(a \, da_1 \wedge \dotsb \wedge da_r) = da \wedge da_1 \wedge \dotsb \wedge da_r$
for $a, a_1, \dotsc, a_r \in A$. It satisfies $d^{r+1}\circ d^r=0$ and hence
defines a complex of $R$-modules
\[
	\Om_{A/R}^\cdot\colon \quad 0 \rightarrow A \stackrel{d}{\rightarrow} \Om_{A/R}^1 \stackrel{d^1}{\rightarrow} 
	\dotsb \rightarrow \Om_{A/R}^r \stackrel{d^r}{\rightarrow} \Om_{A/R}^{r+1} \rightarrow \dotsb
\]
called the \emph{de Rham complex}  of $A$ over $R$. This complex also has an
$R$-algebra structure with the exterior product.
The K\"ahler differentials satisfy the following properties, which make it
convenient to study algebra extensions.

\begin{lem}[Base change] \label{lem:DeRhamBaseChange}
	Let $R$ be a ring, let $A$ and $R'$ be $R$-algebras. Then $A' := R' \otimes_R A$ is an $R'$-algebra and,
	for all $r \ge 0$, there is an $A'$-module isomorphism
	$R' \otimes_R \Om_{A/R}^r \rightarrow \Om_{A'/R'}^r$
	given by $r' \otimes (da_1 \wedge \dotsb \wedge da_r) \mapsto (r'\otimes 1) \, d(1 \otimes a_1) \wedge \dotsb \wedge d(1 \otimes a_r)$.
\end{lem}

\begin{lem}[Localization] \label{lem:DeRhamLocalization}
	Let $R$ be a ring, let $A$ be an $R$-algebra and let $B = S^{-1}A$ for some
	multiplicatively closed set $S \subset A$. Then there 
	is a $B$-module isomorphism
	$B \otimes_A \Om_{A/R}^r \rightarrow \Om_{B/R}^r$
	given by $b \otimes (da_1 \wedge \dotsb \wedge da_r) \mapsto b \, da_1 \wedge \dotsb \wedge da_r$.
	The universal $R$-derivation $d \colon B \rightarrow \Om_{B/R}^1$ satisfies $d(s^{-1}) = -s^{-2} \, ds$ for $s \in S$.
\end{lem}

For $r = 1$ these lemmas are proved in~\cite{bib:Eis95} as Propositions 16.4 and
16.9, respectively, and for $r\ge 2$ they follow from~\cite[Proposition A2.2~b]{bib:Eis95}.

The Jacobian emerges quite naturally in this setting.

\begin{defn} \label{defn:JacobianDifferential}
	The \emph{Jacobian differential of $\term{a}_r\in A$} is defined as
	$\J_{A/R}(\term{a}_r) := d a_1 \wedge \dotsb \wedge d a_r \in \Om_{A/R}^r$.
\end{defn}

Now consider the polynomial ring $k[\term{x}]$.
Then $\Om_{k[\term{x}]/k}^1$ is a free $k[\term{x}]$-module of rank $n$ with basis
$dx_1, \dotsc, dx_n$. It follows that
$\Om_{k[\term{x}]/k}^r = 0$ for $r>n$. For $r \le n$ and
$I = \{j_1 < \dotsb < j_r\} \in \tbinom{[n]}{r}$, we use the notation
$\bigwedge_{j \in I} d x_j := d x_{j_1} \wedge \dotsb \wedge d x_{j_r}$.
The $k[\term{x}]$-module $\Om_{k[\term{x}]/k}^r$ is free of rank~$\tbinom{n}{r}$
with basis 
$\bigl\{ \bigwedge_{j \in I} d x_j \,\vert\; I \in \tbinom{[n]}{r} \bigr\}$.
The derivation $d \colon k[\term{x}] \rightarrow \Om_{k[\term{x}]/k}^1$
is given by $f \mapsto \sum_{i=1}^n (\partial_{x_i}f) dx_i$.

The \emph{Jacobian matrix} of $\term{f}_m \in k[\term{x}]$ is
$\cJ_{\term{x}}(\term{f}_m) := (\partial_{x_j} f_i)_{i,j} \in k[\term{x}]^{m \times n}$.
For an index set 
$I = \{j_1 < \dotsb < j_r\} \in \tbinom{[n]}{r}$, we write $\term{x}_I:=(x_{j_1},\dotsc,x_{j_r})$
and
$\cJ_{\term{x}_I}(\term{f}_m) := (\partial_{x_{j_k}} f_i)_{i,k} \in k[\term{x}]^{m \times r}$.
A standard computation shows
\[
	d f_1 \wedge \dotsb \wedge d f_r = \sideset{}{_I}\sum \det \cJ_{\term{x}_I}(\term{f}_r) 
			\cdot {\textstyle \bigwedge_{j \in I} d x_j},
\]
where the sum runs over all $I \in \tbinom{[n]}{r}$, which implies the following relationship 
between the Jacobian differential and the rank of the Jacobian matrix.

\begin{lem} \label{lem:RelationshipJacobianDifferentialJacobianMatrix}
	For $\term{f}_r \in k[\term{x}]$ we have
	$\J_{k[\term{x}]/k}(\term{f}_r) \neq 0$ if and only if
	$\rk_{k(\term{x})} \cJ_{\term{x}}(\term{f}_r) = r$. 
\end{lem}

\subsection{Separability} \label{sec:Separability}

A univariate polynomial $f\in k[x]$ is called \emph{separable} if it has no
multiple roots in~$\ol{k}$. If $f$ is irreducible, then it is separable if and
only if $\partial_x f \neq 0$, which is always the case in characteristic zero.
If $\ch(k) = p > 0$, then $f$ is separable if and only if $f \notin k[x^p]$.
Now let $L/k$ be a field extension. An algebraic element $a \in L$ over $k$ is
called \emph{separable} if its minimal polynomial
in $k[x]$ is separable. The separable elements form a field $k \subseteq k_{\sep} \subseteq L$ which
is called the \emph{separable closure} of $k$ in $L$.
Now let $L/k$ be an algebraic extension. Then $[L:k]_{\sep} := [k_{\sep}:k]$ resp.
$[L:k]_{\insep} := [L:k_{\sep}]$
are called \emph{separable} resp. \emph{inseparable degree} of $L/k$.
If $L = k_{\sep}$, then $L/k$ is called \emph{separable}.
The extension $L/k_{\sep}$ is \emph{purely inseparable}, i.e. $a^{p^e} \in k_{\sep}$
for some $e \ge 0$, where $p = \ch(k)$.

More generally, a finitely generated extension
$L/k$ is \emph{separable} if it has a transcendence basis $B \subset L$ such
that the finite extension~$L/k(B)$ is separable. In this case, $B$ is called a
\emph{separating transcendence basis} of $L/k$. If $L/k$ is separable, then
every generating system of~$L$ over $k$ contains a separating transcendence
basis. If $k$ is perfect, then every finitely generated field extension of $k$
is separable~\cite[\S X.6]{bib:Lang84}.

Lemma 16.15 in~\cite{bib:Eis95} implies that a separable field extension adds
no new linear relations in the differential module, and Proposition A2.2~b [loc.cit.] yields

\begin{lem}[Separable extension] \label{lem:DeRhamFiniteSeparableExtension}
	Let $L/k$ be a separable algebraic field extension and let $R$ be a subring of $k$.
	Then there is an $L$-vector space isomorphism
	$L \otimes_k \Om_{k/R}^r \cong \Om_{L/R}^r$
	given by $b \otimes (da_1 \wedge \dotsb \wedge da_r) \mapsto b\,da_1 \wedge \dotsb \wedge da_r$.
\end{lem}

\subsection{The ring of Witt vectors} \label{sec:WittVectors}

The Witt ring was defined in~\cite{bib:Wit36}. For its precise definition and
basic properties we also refer to~\cite{bib:Lang84,bib:Ser79,bib:Haz78}.

Fix a prime $p$ and a ring $A$. As a set, the \emph{ring $\W(A)$ of ($p$-typical) Witt vectors of $A$}
(or {\em Witt ring} for short)
is defined as $A^{\N}$. An element $a \in \W(A)$ is written
$(a_0, a_1, \dotsc)$ and is called a \emph{Witt vector} with \emph{coordinates} $a_i \in A$.
The ring structure of $\W(A)$ is given by universal polynomials $S_i, P_i \in \Z[x_0, \dotsc, x_i, y_0, \dotsc, y_i]$
such that
\[
	a+b = (S_0(a_0, b_0), S_1(a_0,a_1,b_0,b_1), \dotsc ), \quad
	ab = (P_0(a_0, b_0), P_1(a_0,a_1,b_0,b_1), \dotsc )
\]
for all $a, b \in \W(A)$. The first few terms are $S_0 = x_0+y_0$, $P_0 = x_0 y_0$,
\[
	S_1 = x_1 + y_1 - {\textstyle\sum_{i=1}^{p-1} p^{-1} \tbinom{p}{i} x_0^i y_0^{p-i}}, \quad
	P_1 = x_0^p y_1+ x_1 y_0^p + p x_1 y_1 .
\]
The additive and multiplicative identity elements of $\W(A)$ are $(0, 0, 0, \dotsc)$ and $(1, 0, 0, \dotsc)$,
respectively. The ring structure is uniquely determined by a universal property,
which we refrain from stating.
If $p$ is invertible in $A$, then $\W(A)$ is isomorphic to $A^\N$ with componentwise
operations.

The projection $\W_\ell(A)$ of $\W(A)$ to the first $\ell\ge 1$ coordinates is a ring
with the same rules for addition and multiplication as for $\W(A)$, which is
called the \emph{ring of Witt vectors of $A$ of length $\ell$}. We have $\W_1(A) = A$. The
ring epimorphisms
\[
	\Restr\colon \W_{\ell+1}(A) \rightarrow \W_\ell(A), \qquad (a_0, \dotsc, a_\ell) \mapsto (a_0, \dotsc, a_{\ell-1})
\]
are called \emph{restriction} and $((\W_\ell(A))_{\ell \ge 1}, \Restr\colon \W_{\ell+1}(A) \rightarrow \W_\ell(A))$ is a 
projective (inverse) system of rings with limit
$\W(A)$. The additive group homomorphism
\[
	\V\colon \W(A) \rightarrow \W(A), \qquad (a_0, a_1, \dotsc ) \mapsto (0, a_0, a_1, \dotsc)
\]
is called \emph{Verschiebung} (shift). For $\ell, r \ge 1$, we have exact sequences
\[
	0 \rightarrow \W(A) \stackrel{\V^\ell}{\rightarrow} \W(A) \rightarrow \W_\ell(A) \rightarrow 0, \quad
	0 \rightarrow \W_r(A) \stackrel{\V^\ell}{\rightarrow} \W_{\ell+r}(A) 
	\stackrel{\Restr^r}{\rightarrow} \W_\ell(A) \rightarrow 0.
\]
The Verschiebung also induces additive maps $\V\colon\W_\ell(A)\to\W_{\ell+1}(A)$.

The \emph{Teichm\"{u}ller lift} of $a \in A$ is defined as $[a] := (a, 0, 0, \dotsc) \in \W(A)$.
The image of $[a]$ in $\W_\ell(A)$ is denoted by $[a]_{\le\ell}$. We have
\[
	[a] \cdot w = \bigl(a w_0, a^p w_1, \dotsc, a^{p^i} w_i, \dotsc\bigr)
\]
for all $w \in \W(A)$. In particular, the map $A \rightarrow \W(A)$, $a \mapsto [a]$ is multiplicative,
i.\,e., $[ab] = [a][b]$ for all $a, b \in A$. Every $a \in \W(A)$ can be written as
$a = \sum_{i=0}^\infty \V^i[a_i]$.

We are only interested in the case where $A$ has characteristic $p$.
The most basic example is the prime field $A=\F_p$, for which $\W(\F_p)$ is
the ring $\hat{\Z}_p$ of $p$-adic integers. More generally, the Witt ring $\W(\F_{p^t})$
of a finite field~$\F_{p^t}$ is the ring of integers $\hat{\Z}_p^{(t)}$ in the
unique unramified extension $\Q_p^{(t)}$ of $\Q_p$ of degree $t$ \cite{bib:Kob84}.

Now let $A$ be an $\F_p$-algebra.
Then the \emph{Frobenius endomorphism} $\Frob\colon A \rightarrow A$, $a \mapsto a^p$ induces a ring endomorphism
\begin{equation}\label{eq:frobWittRing}
	\Frob\colon \W(A) \rightarrow \W(A), \qquad (a_0, a_1, \dotsc) \mapsto (a_0^p, a_1^p, \dotsc).
\end{equation}
We have $\V\Frob = \Frob\V = p$ and $a \V b = \V(\Frob a \cdot b)$ for all $a, b \in \W(A)$.
The Frobenius further induces endomorphisms on $\W_\ell(A)$.
An $\F_p$-algebra $A$ is called {\em perfect}, if $\Frob$ is an automorphism. In
this case, the induced endomorphism $\Frob$ on $\W(A)$ is an automorphism as well.

%

Let $v_p \colon \Q \rightarrow \Z \cup \{\infty\}$ denote the \emph{$p$-adic valuation}
of $\Q$. For a nonzero $q \in \Q$, $v_p(q)$ is defined as the unique integer $v \in \Z$
such that $q = p^v \frac{a}{b}$ for $a, b \in \Z\setminus p\Z$.
For tuples
$\alpha \in \Q^s$, $s \ge 1$, set $v_p(\alpha) := \min_{1 \le i \le s} v_p(\alpha_i) \in \Z \cup \{\infty\}$.

\begin{lem}[Expanding Teichm\"{u}ller] \label{lem:TeichmuellerExpansion}
	Let $A = R[\term{a}] = R[\term{a}_n]$ be an $R$-algebra, where $R$ is an $\F_p$-algebra,
	and let $f = \sum_{i=1}^s c_i \term{a}^{\alpha_i} \in A$, where $c_i \in R$ and $\alpha_i \in \N^n$. 
	Then, in $\W_{\ell+1}(A)$, we have the sum over $\term{i} \in \N^s$ and 
	$\tbinom{p^\ell}{\term{i}} = \tbinom{p^\ell}{i_1, \dotsc, i_s}$ :
	\begin{equation} \label{eqn:TeichmuellerExpansion}
		[f] = \sum_{\abs{\term{i}}=p^\ell} p^{-\ell+v_p(\term{i})}
		\tbinom{p^\ell}{\term{i}} \cdot \V^{\ell-v_p(\term{i})} \Frob^{-v_p(\term{i})}
		\bigl([c_1 \term{a}^{\alpha_1}]^{i_1} \dotsb [c_s \term{a}^{\alpha_s}]^{i_s} \bigr).
	\end{equation}
\end{lem}
\begin{proof}
	Note that the RHS $w$ of \eqref{eqn:TeichmuellerExpansion} is a well-defined element of $\W(A)$, 
	because $p^{-\ell+v_p(\term{i})} \cdot \tbinom{p^\ell}{\term{i}} \in \N$
	by Lemma \ref{lem:DivisibilityMultinomialCoefficient}, $v_p(\term{i}) \le \ell$ and 
	$p^{-v_p(\term{i})} \cdot \term{i} \in \N^s$. We have $[f] = \sum_{i=1}^s [c_i \term{a}^{\alpha_i}]$
	in $\W_1(A)$, so Lemma \ref{lem:PthPowerAndFiltration} implies
	\[
		\Frob^\ell [f] = [f]^{p^\ell} = \bigl( {\textstyle\sum_{i=1}^s [c_i \term{a}^{\alpha_i}]}\bigr)^{p^\ell}
		= \sum_{\abs{\term{i}}=p^\ell} \tbinom{p^\ell}{\term{i}} \cdot
		[c_1 \term{a}^{\alpha_1}]^{i_1} \dotsb [c_s \term{a}^{\alpha_s}]^{i_s} \quad\text{in $\W_{\ell+1}(A)$}.
	\]
	Since $\V\Frob = \Frob\V = p$, we see that this is equal to $\Frob^\ell w$. The injectivity of $\Frob$
	implies $[f] = w$ in $\W_{\ell+1}(A)$.
\end{proof}

\subsection{The de Rham-Witt complex} \label{sec:DeRhamWittComplex}

For this section we refer to~\cite{bib:Il79}.
Let $R$ be a ring. Recall that a {\em differential graded $R$-algebra} ({\em $R$-dga} for short) is a graded $R$-algebra
$M=\bigoplus_{r\ge 0}M^r$ together with an $R$-linear differential $d\colon M^r\to M^{r+1}$ such that
$M$ is graded skew-commutative, i.e., $ab=(-1)^{rs}ba$ for $a\in M^r,b\in M^s$
(in fact, we also assume that $a^2=0$ for $a\in M^{2r+1}$), and
$d$ satisfies: $d\circ d=0$ and the graded Leibniz rule $d(ab)=b\ da+(-1)^ra\ db$ for $a\in M^r$, $b\in M$.
A $\Z$-dga is simply called {\em dga}.
An important example is the $R$-dga $\Om_{A/R}^\cdot:=\bigoplus_{r\ge 0}\Om_{A/R}^r$
together with $d:=\bigoplus_{r\ge 0}d^r$.

\begin{defn} \label{def:DeRhamWittComplex}
	Fix a prime $p$.
	A \emph{de Rham $\V$-pro-complex} ({\em VDR} for short) is a projective
        system $M_{\bullet} = ((M_\ell)_{\ell\ge 1}, \Restr\colon M_{\ell+1} \rightarrow M_\ell)$
	of dga's together with additive homomorphisms
	$(\V\colon M_\ell^r \rightarrow M_{\ell+1}^r)_{r\ge0,\ell\ge 1}$ such that $\Restr\V = \V\Restr$
        and we have
	\begin{enumerate}[(a)]
		\item 
			$M_1^0$ is an $\F_p$-algebra and $M_\ell^0 = \W_\ell(M_1^0)$ with the restriction and Verschiebung maps of Witt rings
			$\Restr\colon M_{\ell+1}^0 \rightarrow M_{\ell}^0$ and $\V\colon M_\ell^0 \rightarrow M_{\ell+1}^0$,

		\item $\V(\omega \, d\eta) = (\V\omega) d\V\eta$ for
			all $\omega \in M_\ell^r, \eta \in M_\ell^s$,
		\item $(\V w) d[a] = \V([a]^{p-1} w) d \V [a]$ for all
			$a \in M_1^0, w \in M_\ell^0$.
	\end{enumerate}
\end{defn}

\cite{bib:Il79} constructs for any $\F_p$-algebra $A$ a functorial de Rham $\V$-pro-complex
$\W_\bullet\Om_A^\cdot$ with $\W_\ell\Om_A^0=\W_\ell(A)$, which is called the
\emph{de Rham-Witt pro-complex of~$A$}. We have a surjection $\Om_{\W_\ell(A)/\W_\ell(\F_p)}^\cdot 
\twoheadrightarrow \W_\ell\Om_A^\cdot$,
which restricts to the identity on $\W_\ell(A)$ and, for $\ell=1$, is an isomorphism
$\Om_{\W_1(A)/\F_p}^\cdot=\Om_{A/\F_p}^\cdot \stackrel{\sim}{\to} \W_1\Om_A^\cdot$.

Like the K\"ahler differentials, $\W_\bullet\Om_A^\cdot$ satisfy properties that 
make it convenient to study algebra extensions.

\begin{lem}[Base change {\cite[Proposition I.1.9.2]{bib:Il79}}] \label{lem:DeRhamWittBaseChange}
	Let $k'/k$ be an extension of perfect fields of characteristic $p$. Let $A$ be
	a $k$-algebra and set $A' := k' \otimes_k A$. Then there is a natural $\W_\ell(k')$-module isomorphism 
	$\W_\ell(k') \otimes_{\W_\ell(k)} \W_\ell\Om_A^r \cong \W_\ell\Om_{A'}^r$ for all $\ell \ge 1$ and $r \ge 0$.
\end{lem}
\begin{lem}[Localization {\cite[Proposition I.1.11]{bib:Il79}}] \label{lem:DeRhamWittLocalization}
	Let $A$ be an $\F_p$-algebra and let $B = S^{-1}A$ for some
	multiplicatively closed set $S \subset A$. Then there is a natural
	$\W_\ell(B)$-module isomorphism 
	$\W_\ell(B) \otimes_{\W_\ell(A)} \W_\ell\Om^r_A \cong \W_\ell\Om^r_B$ for all $\ell \ge 1$ and $r \ge 0$.
\end{lem}

\begin{lem}[Separable extension] \label{lem:DeRhamWittFiniteSeparableExtension}
	Let $L/K$ be a finite separable field extension of characteristic $p$.
	Then there is a natural $\W_\ell(L)$-module isomorphism
	$\W_\ell(L) \otimes_{\W_\ell(K)} \W_\ell\Om^r_K \cong \W_\ell\Om^r_L$ for all $\ell \ge 1$ and $r \ge 0$.
\end{lem}
\begin{proof}
Proposition I.1.14 of~\cite{bib:Il79} states this for an \'{e}tale morphism $K\rightarrow L$,
which means flat and unramified. A vector space over a field is immediately flat,
and a finite separable field extension is unramified by definition (see e.g.~\cite{bib:Milne80}).
\end{proof}
\begin{rem}
The proofs in~\cite{bib:Il79} show that the isomorphisms of Lemmas~\ref{lem:DeRhamWittBaseChange}
-- \ref{lem:DeRhamWittFiniteSeparableExtension} are in fact isomorphisms of VDR's
with appropriately defined VDR-structures.
\end{rem}

According to~\cite[Th\'{e}or\`{e}me I.2.17]{bib:Il79}, the morphism of projective
systems of rings $\Restr\Frob=\Frob \Restr\colon\W_{\bullet}(A)\to\W_{\bullet-1}(A)$
uniquely extends to a morphism of projective systems of graded algebras
$\Frob\colon \W_\bullet\Om_A^\cdot\to\W_{\bullet-1}\Om_A^\cdot$ such that
$\Frob d[a]_{\le\ell+1}=[a]_{\le\ell}^{p-1} d[a]_{\le\ell}$ for all $a\in A$, and
$\Frob d\V=d$ in $\W_\ell\Om_A^1$ for all $\ell\ge 1$.
Define the {\em canonical filtration} as
$\Fil^\ell\W_m\Om_A^\cdot:=\ker\big(\Restr^{m-\ell}\colon\W_m\Om_A^\cdot\to\W_{\ell}\Om_A^\cdot)$ for $\ell,m\ge 0$.

Now consider a function field $L:=k(\term{x}_n)$ over a perfect field $k$. The following fact, proven in 
\S\ref{app:Preliminaries}, is quite useful for our differential calculations.

\begin{lem}[Frobenius kernel]\label{lem:kerFrob}
We have $\ker\big(\W_{\ell+i}\Om_L^r\stackrel{\Frob^i}{\longrightarrow}\W_{\ell}\Om_L^r\big)\subseteq\Fil^\ell\W_{\ell+i}\Om_L^r$.
\end{lem}

\subsection{The de Rham-Witt complex of a polynomial ring} \label{sec:DeRhamWittComplexPolynomialRing}

Let $k/\F_p$ be an algebraic extension and
consider the polynomial ring $A := k[\term{x}] = k[\term{x}_n]$.
In~\cite[\S I.2]{bib:Il79} there is an explicit description of $\W_\bullet\Om_A^\cdot$ in the
 case $k = \F_p$.
We generalize this construction by invoking Lemma \ref{lem:DeRhamWittBaseChange}
(note that $k$ is perfect).

Denote by $K := Q(\W(k))$ the quotient field of the Witt ring, and consider
the rings $B := \W(k)[\term{x}]$ and 
$C := \bigcup_{i \ge 0} K[\term{x}^{p^{-i}}]$. 
For $r \ge 0$, we write $\Om_B^r := \Om_{B/\W(k)}^r$ and $\Om_C^r := \Om_{C/K}^r$.
Since the universal derivation $d \colon C \rightarrow \Om_C^1$ satisfies
\[
	d\bigl(x_j^{p^{-i}}\bigr) = p^{-i} x_j^{p^{-i}} dx_j/x_j\quad\text{for all}\quad i \ge 0,\ j\in[n],
\]
every differential form $\omega \in \Om_C^r$ can be written uniquely as
\begin{equation}\label{eq:generalForm}
	\omega = \sideset{}{_I}\sum c_I \cdot {\textstyle \bigwedge_{j \in I} d \log x_j},
\end{equation}
where the sum is over all $I \in \tbinom{[n]}{r}$, the $c_I \in C$ are divisible
by $(\prod_{j \in I} x_j)^{p^{-s}}$ for some $s \ge 0$, and $d \log x_j := dx_j/x_j$.
The $c_I$ in~\eqref{eq:generalForm} are called \emph{coordinates} of~$\omega$. 
A form $\omega$ is called {\em integral} if all its coordinates 
have coefficients in $\W(k)$.
We define
\[
	\E^r := \E_A^r := \{ \omega \in \Om_C^r \,\vert\; \text{both $\omega$ and $d\omega$ are integral}\}.
\]
Then, $\E:=\bigoplus_{r\ge 0}\E^r$ is a differential graded subalgebra of $\Om_C$ containing $\Om_B$.

Let $\Frob\colon C \rightarrow C$ be the unique $\Q_p$-algebra automorphism
extending the Frobenius of $\W(k)$ defined by~\eqref{eq:frobWittRing} and
sending $x_j^{p^{-i}}$ to $x_j^{p^{-i+1}}$.
The map $\Frob$ extends to an automorphism $\Frob\colon \Omega_C^r \rightarrow \Omega_C^r$ of dga's
by acting on the coordinates of the differential forms (keeping $d\log x_j$ fixed), and we define 
$\V\colon \Omega_C^r \rightarrow \Omega_C^r$
by $\V := p \Frob^{-1}$. We have $d \Frob = p \Frob d$ and $\V d = p d \V$, in particular,
$\E$ is closed under $\Frob$ and $\V$.

We define a filtration $\E = \Fil^0\E \supset \Fil^1\E \supset \dotsb$ of differential
graded ideals by
\[
	\Fil^\ell \E^r := \V^\ell \E^r + d \V^\ell \E^{r-1}\quad\text{for}\quad\ell,r \ge 0,
\]
and hence obtain a projective system $\E_\bullet$ of dga's
\[
	\E_\ell := \E/\Fil^\ell \E,\quad \Restr\colon\E_{\ell+1}\twoheadrightarrow\E_\ell.
\]

\begin{thm}[Explicit forms] \label{thm:DeRhamWittOfPolynomialRing}
	The system $\E_\bullet$ is a VDR, isomorphic to $\W_\bullet\Om_A^\cdot$.
\end{thm}
\begin{proof}
The case $k=\F_p$ follows from \cite[Th\'{e}or\`{e}me I.2.5]{bib:Il79}. Lemma~\ref{lem:DeRhamWittBaseChange}
yields $\W_\bullet\Om_A^\cdot \cong \W_\bullet(k) \otimes_{\W(\F_p)} \W_\bullet\Om_{\F_p[\term{x}]}^\cdot$
as VDR's. In particular, the Verschiebung restricts to the
Verschiebung of $\W_\bullet(A)$, so it coincides with the map $\V$ defined above.
\end{proof}

\begin{lem}[{\cite[Corollaire I.2.13]{bib:Il79}}] \label{lem:DeRhamWittPTorsion}
	Multiplication with $p$ in $E$ induces for all $\ell \ge 0$ a well-defined injective map
        $m_p\colon \E_\ell \rightarrow \E_{\ell+1}$ with $m_p\circ\Restr=p$.
\end{lem}

\section{The classical Jacobian criterion} \label{sec:SeparabilityClassicalJacobianCriterion}

Consider a polynomial ring $k[\term{x}] = k[\term{x}_n]$. In this section we
characterize the zeroness of the Jacobian differential which, combined with Lemma 
\ref{lem:RelationshipJacobianDifferentialJacobianMatrix}, gives a criterion on the Jacobian
matrix. The proofs for this section can be found in \S\ref{app:ProofsClassicalJacobianCriterion}.

\begin{thm}[Jacobian criterion -- abstract] \label{thm:JacobianCriterion}
	Let $\term{f}_r \in k[\term{x}]$ be polynomials. Assume that $k(\term{x})$ is 
	a separable extension of $k(\term{f}_r)$.
	Then, $\term{f}_r$ are algebraically independent over $k$ if and only if
	$\J_{k[\term{x}]/k}(\term{f}_r) \neq 0$.
\end{thm}


As a consequence of Theorem \ref{thm:PerronsTheorem}, the separability hypothesis of 
Theorem~\ref{thm:JacobianCriterion} is satisfied in sufficiently large characteristic.

\begin{lem} \label{lem:SeparabilityInLargeChar}
	Let $\term{f}_m \in k[\term{x}]$ have transcendence degree $r$ and
	maximal degree $\delta$, and assume $\ch(k) = 0$ or $\ch(k) > \delta^r$. Then the extension
	$k(\term{x})/k(\term{f}_m)$ is separable.
\end{lem}

\section{The Witt-Jacobian criterion} \label{sec:WJCriterion}

This we prove in two steps. First, an abstract criterion (zeroness of a differential). Second, an explicit 
criterion (degeneracy of a $p$-adic polynomial).

\subsection{The Witt-Jacobian differential} \label{sec:WJDifferential}

\begin{defn} \label{defn:WittJacobianDifferential}
	Let $A$ be an $\F_p$-algebra, $\term{a}_r \in A$, and $\ell \ge 1$.
	We call $\WJ_{\ell,A}(\term{a}_r) := d [a_1]_{\le\ell} \wedge \dotsb \wedge d [a_r]_{\le\ell} \in \W_\ell\Om^r_A$
	the \emph{($\ell$-th) Witt-Jacobian differential of $\term{a}_r$ in $\W_\ell\Om^r_A$}. 
\end{defn}

Let $k$ be an algebraic extension field of $\F_p$ (thus, $k \subseteq \ol{\F}_p$).

\begin{lem} \label{lem:DiffModuleZero}
	Let $L/k$ be a finitely generated field extension and let $\ell \ge 1$.
	Then $\W_{\ell}\Om^r_L = 0$ if and only if $r> \trdeg_k(L)$.
\end{lem}
\begin{proof}
	Let $s := \trdeg_k(L)$. Since $L$ is finitely generated over a perfect field, 
	it has a separating transcendence basis $\{a_1, \dotsc, a_s\} \subset L$. This means that
	$L$ is a finite separable extension of $K := k(\term{a}_s)$.
        Since $A := k[\term{a}_s]$ is isomorphic to a polynomial ring over
        $k$, we have $\W_\ell\Om_A^r = 0$ iff $r \ge s+1$
        by~\S\ref{sec:DeRhamWittComplexPolynomialRing}.
	Lemmas \ref{lem:DeRhamWittLocalization} and \ref{lem:DeRhamWittFiniteSeparableExtension}
        imply $\W_\ell\Om_A^r = 0$ iff $\W_\ell\Om_K^r = 0$  iff $\W_{\ell}\Om_L^r = 0$.
\end{proof}

\begin{cor} \label{cor:DiffModuleZero}
	For an affine $k$-domain $A$ and $\ell \ge 1$,
	$\W_{\ell}\Om_A^r = 0$ iff $r> \trdeg_k(A)$.
\end{cor}
\begin{proof}
Apply Lemma \ref{lem:DiffModuleZero} to the quotient field
of $A$ and use Lemma \ref{lem:DeRhamWittLocalization}.
\end{proof}

Now let $A := k[\term{x}] = k[\term{x}_n]$ be a polynomial ring over $k$.

\begin{lem}[Zeroness] \label{lem:AbstractWJCriterion1}
	If $\term{f}_r \in A$ are algebraically dependent, then 
	$\WJ_{\ell, A}(\term{f}_r) = 0$ for all $\ell \ge 1$.
\end{lem}
\begin{proof}
	Assume that $\term{f}_r$ are algebraically dependent and set
        $R := k[\term{f}_r]$.
        Corollary~\ref{cor:DiffModuleZero} implies $\W_{\ell}\Om^r_R = 0$,
        thus $\WJ_{\ell, R}(\term{f}_r) = 0$.
	The inclusion $R \subseteq A$ induces a homomorphism 
	$\W_{\ell}\Om_R^r \to \W_{\ell}\Om_A^r$, hence 
	$\WJ_{\ell, A}(\term{f}_r) = 0$.
\end{proof}

We extend the \emph{inseparable degree} to finitely generated field 
extensions $L/K$ by 
$[L:K]_{\insep} := \min \bigl\{ [L:K(B)]_{\insep} \,\vert\;
	\text{$B \subset L$ is a transcendence basis of $L/K$} \bigr\}$.
Note that $[L:K]_{\insep}$ is a power of $\ch(K)$, and equals $1$
iff $L/K$ is separable.

\begin{lem}[Non-zeroness] \label{lem:AbstractWJCriterion2}
	If $\term{f}_r \in A$ are algebraically independent, then we have $\WJ_{\ell, A}(\term{f}_r) \neq 0$
	for all $\ell> \log_p [k(\term{x}):k(\term{f}_r)]_{\insep}$.
\end{lem}
\begin{proof}
	It suffices to consider the case $\ell = e+1$, where $e := \log_p [k(\term{x}):k(\term{f}_r)]_{\insep}$.
	By definition of $e$, there exist $\term{f}_{[r+1,n]} \in k(\term{x})$ such that
	$L := k(\term{x})$ is algebraic over $K = k(\term{f}) := k(\term{f}_n)$ with
	$[L:K]_{\insep} = p^e$. Let $K_{\sep}$ be the separable closure of $K$
	in $L$, thus $L/K_{\sep}$ is purely inseparable. For $i \in [0,n]$,
	define the fields $K_i := K_{\sep}[x_1, \dotsc, x_i]$, hence we have a tower
	$K \subseteq K_{\sep} = K_0 \subseteq K_1 \subseteq \dotsb \subseteq K_n= L$.
	For $i \in [n]$, let $e_i \ge 0$ be minimal such that $x_i^{p^{e_i}} \in K_{i-1}$
	($e_i$ exists, since $K_i/K_{i-1}$ is purely inseparable). Set $q_i := p^{e_i}$.
	By the multiplicativity of field extension degrees, we have $e = \sum_{i=1}^n e_i$.
	
	Since $\WJ_{1,A}(\term{x}) \neq 0$, we have 
	$p^e \cdot \WJ_{\ell,A}(\term{x}) = m_p^e \WJ_{1,A}(\term{x}) \neq 0$ by
	Lemma~\ref{lem:DeRhamWittPTorsion}. Lemma~\ref{lem:DeRhamWittLocalization}
	implies $p^e \cdot \WJ_{\ell,L}(\term{x}) \neq 0$. We conclude
	\begin{equation} \label{eqn:LemmaAbstractWJCriterion2}
		\WJ_{\ell,L}(x_1^{q_1}, \dotsc, x_n^{q_n})
		= p^e \cdot [x_1]^{q_1-1} \dotsb [x_n]^{q_n-1} \cdot
		\WJ_{\ell,L}(\term{x}) \neq 0,
	\end{equation}
	since $[x_1]^{q_1-1} \dotsb [x_n]^{q_n-1}$ is a unit in $\W_\ell(L)$.
	
	Now assume for the sake of contradiction that $\WJ_{\ell,L}(\term{f}) = 0$.
	We want to show inductively for $j = 0, \dotsc, n-1$ that the induced map
	$\Psi_j\colon \W_\ell\Om_{K_j}^n \rightarrow \W_\ell\Om_{L}^n$ satisfies
	\[
		\Psi_j\bigl(d[x_1^{q_1}] \wedge \dotsb \wedge d[x_j^{q_j}] \wedge 
		d[a_{j+1}] \wedge \dotsb \wedge d[a_n]\bigr) = 0 \quad
		\text{for all $a_{j+1}, \dotsc, a_n \in K_j$}.
	\]
	To prove this claim for $j = 0$, we first show that, for $R := k[\term{f}]$, the induced map
	$\Psi\colon\W_\ell\Om_R^n \rightarrow \W_\ell\Om_{L}^n$ is zero. By Lemma \ref{lem:TeichmuellerExpansion},
	every element $\ol{\omega} \in \W_\ell\Om_R^n$ is a $\Z$-linear combination of products of elements
	of the form $\V^i[c \term{f}^{\alpha}]$ and $d\V^i[c \term{f}^{\alpha}]$ for some
	$i \in [0, \ell-1]$, $c \in k$, and $\alpha \in \N^n$. Wlog., let
	$\ol{\omega} = \V^{i_0}[c_0 \term{f}^{\alpha_0}] \cdot
	d\V^{i_1}[c_1 \term{f}^{\alpha_1}] \wedge \dotsb \wedge d\V^{i_n}[c_n \term{f}^{\alpha_n}]$.
	Let $\omega \in \W_m\Om_R^n$ be a lift of $\ol{\omega}$ for $m$ sufficiently large (say $m=2\ell$).
	Using $\Frob d \V = d$ and $\Frob d [w] = [w]^{p-1} d[w]$ for $w \in R$, we deduce
	$\Frob^{\ell}\omega = g \cdot d[c_1 \term{f}^{\alpha_1}] \wedge \dotsb \wedge d[c_n \term{f}^{\alpha_n}]$
	for some $g \in \W_{m-\ell}(R)$. By the Leibniz rule, we can simplify to
	$\Frob^{\ell}\omega = g' \cdot d[f_1] \wedge \dotsb \wedge d[f_n]$ for some $g' \in \W_{m-\ell}(R)$.
	Since $\WJ_{\ell,L}(\term{f}) = 0$ by assumption, we obtain 
	$\Frob^{\ell}\Psi(\omega) = \Psi(\Frob^{\ell}\omega) \in \Fil^\ell\W_{m-\ell}\Om_L^n$,
	hence $\Psi(\omega) \in \Fil^\ell\W_m\Om_L^n$ by Lemma \ref{lem:kerFrob}. This shows
	$\Psi(\ol{\omega}) = 0$, so $\Psi$ is zero. Lemmas \ref{lem:DeRhamWittLocalization}
	and \ref{lem:DeRhamWittFiniteSeparableExtension} imply that the map
	$\Psi_0$ is zero, proving the claim for $j=0$.
	
	Now let $j \ge 1$ and let $\ol{\omega} = d[x_1^{q_1}] \wedge \dotsb \wedge d[x_j^{q_j}] \wedge 
	d[a_{j+1}] \wedge \dotsb \wedge d[a_n] \in \W_\ell\Om_{K_j}^n$ with $a_{j+1}, \dotsc, a_n \in K_j$.
	Since $K_j = K_{j-1}[x_j]$, we may assume by Lemma \ref{lem:TeichmuellerExpansion} that
	$\ol{\omega} = d[x_1^{q_1}] \wedge \dotsb \wedge d[x_j^{q_j}] \wedge 
	d\V^{i_{j+1}}[c_{j+1} x_j^{\alpha_{j+1}}] \wedge \dotsb \wedge d\V^{i_n}[c_n x_j^{\alpha_n}]$ with
	$i_{j+1}, \dotsc, i_n \in [0, \ell-1]$, $c_{j+1}, \dotsc, c_n \in K_{j-1}$, and
	$\alpha_{j+1}, \dotsc, \alpha_n \ge 0$. Let $\omega \in \W_m\Om_{K_j}^n$ be a lift of $\ol{\omega}$ for $m$ 
	sufficiently large (say $m=2\ell$). As above, we deduce $\Frob^{\ell}\omega = g \cdot d[x_1^{q_1}] \wedge \dotsb \wedge d[x_j^{q_j}] \wedge 
	d[c_{j+1} x_j^{\alpha_{j+1}}] \wedge \dotsb \wedge d[c_n x_j^{\alpha_n}]$ for some $g \in \W_{m-\ell}(K_j)$,
	and by the Leibniz rule, we can write $\Frob^{\ell}\omega = g' \cdot d[x_1^{q_1}] \wedge \dotsb \wedge d[x_j^{q_j}] \wedge 
	d[c_{j+1}] \wedge \dotsb \wedge d[c_n]$ for some $g' \in \W_{m-\ell}(K_j)$. Since $x_1^{q_1}, \dotsc, x_j^{q_j}$,
	$c_{j+1}, \dotsc, c_n \in K_{j-1}$, we obtain 
	$\Frob^{\ell} \Psi_j(\omega) = \Psi_j(\Frob^{\ell}\omega)\in \Fil^\ell\W_{m-\ell}\Om_L^n$ by induction,
	hence $\Psi_j(\omega) \in \Fil^\ell\W_m\Om_L^n$ by Lemma \ref{lem:kerFrob}. This shows
	$\Psi_j(\ol{\omega}) = 0$, finishing the proof of the claim.
		
	For $j = n-1$ and $a_n = x_n^{q_n} \in K_{n-1}$ the claim implies $\WJ_{\ell,L}(x_1^{q_1}, \dotsc, x_n^{q_n}) = 0$
	which is contradicting \eqref{eqn:LemmaAbstractWJCriterion2}. Therefore, $\WJ_{\ell,L}(\term{f}_n) \neq 0$,
	hence $\WJ_{\ell,L}(\term{f}_r) \neq 0$. Lemma \ref{lem:DeRhamWittLocalization} implies
	$\WJ_{\ell,A}(\term{f}_r) \neq 0$.
\end{proof}

\begin{rem} \label{rem:LowerBoundForEllRefined}
    Lemma \ref{lem:AbstractWJCriterion2} is tight in the case $f_i := x_i^{p^{e_i}}$ for $i \in [r]$.
\end{rem}

\begin{thm}[Witt-Jacobian criterion -- abstract] \label{thm:AbstractWJCriterion}
	Let $\term{f}_r \in A$ be of degree at most $\delta \ge 1$
  and fix $\ell> \lfloor r \log_p \delta \rfloor$. Then,
  $\term{f}_r$ are algebraically independent over $k$ if and only if
	$\WJ_{\ell, A}(\term{f}_r) \neq 0$.
\end{thm}
\begin{proof}
	Let $\term{f}_{[r+1,n]} \subseteq \term{x}$ be a transcendence basis of $k(\term{x})/k(\term{f}_r)$.
	Then $[k(\term{x}):k(\term{f}_r)]_{\insep} \le [k(\term{x}):k(\term{f}_n)]_{\insep}
	\le [k(\term{x}):k(\term{f}_n)] \le \delta^r$ by Theorem \ref{thm:PerronsTheorem}.
	The assertion follows from Lemmas \ref{lem:AbstractWJCriterion1} and
	\ref{lem:AbstractWJCriterion2}.
\end{proof}

\subsection{The Witt-Jacobian polynomial} \label{sec:WJPolynomial}

We adopt the notations and assumptions of~\S\ref{sec:DeRhamWittComplexPolynomialRing}.
In particular, $k/\F_p$ is an algebraic extension, $A = k[\term{x}] = k[\term{x}_n]$,
$B = \W(k)[\term{x}]$, $K = Q(\W(k))$, and $C = \bigcup_{r \ge 0} K[\term{x}^{p^{-r}}]$.
Recall that $\E = \E_{A}$ is a subalgebra of $\Om_C^\cdot$ containing $\Om_B^\cdot$, in
particular, $B \subseteq \E^0$.
Since $k$ is perfect, we have $\W(k)/p\W(k)\cong\W_1(k)=k$ and hence $B/pB\cong A$. 
In the following, we will use these identifications.

\comment{
Let $f \in A$ be a polynomial over $k$ and let $g \in B$ be a polynomial over $R$.
We write $f = g \pmod{pB}$ if $f$ is the image of $g$ under the composition of
the canonical surjection $B \rightarrow B/pB$ and the isomorphism $B/pB \cong A$.
Similarly, for $a \in k$ and $b \in R$, we write $a = b \pmod{pR}$ if $a$ is the
image of $b$ under $R \rightarrow R/pR \cong k$.
For a polynomial $f = \sum_\alpha c_\alpha \term{x}^{\alpha} \in A$, 
$c_\alpha \in k$, we set $\widehat{f} := \sum_\alpha [c_\alpha] \term{x}^{\alpha} \in B$.
Then we have $c_\alpha = [c_\alpha] \pmod{pR}$ and $f = \widehat{f} \pmod{pB}$.
}

\begin{lem}[Realizing Teichm\"uller] \label{lem:TeichmuellerLiftAndFiltration}
	Let $f \in A$ and let $g \in B$ such that $f \equiv g \pmod{pB}$.
	Let $\ell \ge 0$ and let  $\tau\colon \W_{\ell+1}(A) \rightarrow \E_{\ell+1}^0= \E^0/\Fil^{\ell+1}\E^0$ be
	the $\W(k)$-algebra isomorphism from Theorem \ref{thm:DeRhamWittOfPolynomialRing}.  
	Then we have $\tau([f]_{\le\ell+1}) = ( \Frob^{-\ell} g )^{p^\ell}$.
\end{lem}
\begin{proof}
	Write $g = \sum_{i=1}^s c_i \term{x}^{\alpha_i}$, where $c_i \in \W(k)$ and
	$\alpha_i \in \N^n$. By assumption, we have $[f] = \sum_{i=1}^s c_i [\term{x}^{\alpha_i}]$
	in $\W_1(A)$. By Lemma \ref{lem:PthPowerAndFiltration}, we obtain
	\[
		\Frob^\ell [f] = [f]^{p^\ell} = \bigl( {\textstyle\sum_{i=1}^s c_i [\term{x}^{\alpha_i}]}\bigr)^{p^\ell}
		= \sum_{\abs{\term{i}}=p^\ell} \tbinom{p^\ell}{\term{i}} \cdot
		\term{c}^{\term{i}}[\term{x}^{\alpha_1}]^{i_1} \dotsb [\term{x}^{\alpha_s}]^{i_s} \quad\text{in $\W_{\ell+1}(A)$}.
	\]
	As in the proof of Lemma \ref{lem:TeichmuellerExpansion}, this implies
	\[
		[f] = \sum_{\abs{\term{i}}=p^\ell} p^{-\ell+v_p(\term{i})}
		\tbinom{p^\ell}{\term{i}} \cdot \V^{\ell-v_p(\term{i})} \Frob^{-v_p(\term{i})}
		\bigl(c_1^{i_1}[\term{x}^{\alpha_1}]^{i_1} \dotsb c_s^{i_s}[\term{x}^{\alpha_s}]^{i_s} \bigr)
		\quad\text{in $\W_{\ell+1}(A)$}.
	\]
	Since $k$ is perfect, $\Frob$ is an automorphism of $\W(k)$, so this is well-defined.
	Denoting $m_i := c_i \term{x}^{\alpha_i}\in B$, and
	using $\tau\V = \V\tau$ and $\tau([x_i]) = x_i$, we conclude
	\begin{align*}
		\tau([f]) &= \sum_{\abs{\boldsymbol{i}} = p^\ell} p^{-\ell+v_p(\boldsymbol{i})} \tbinom{p^\ell}{\boldsymbol{i}} 
			\V^{\ell - v_p(\boldsymbol{i})} \Frob^{-v_p(\boldsymbol{i})}(m_1^{i_1} \dotsb m_s^{i_s}) \\
		&= \sum_{\abs{\boldsymbol{i}} = p^\ell} \tbinom{p^\ell}{\boldsymbol{i}} 
			\Frob^{-\ell} (m_1^{i_1} \dotsb m_s^{i_s}) 
		= \bigl( {\textstyle \sum_{i=1}^s \Frob^{-\ell} m_i} \bigr)^{p^\ell} = \bigl( \Frob^{-\ell} g \bigr)^{p^\ell}
		\quad\text{in $\E_{\ell+1}^0$}.
	\end{align*}
	Note that the intermediate expression $\Frob^{-\ell} g \in C$ need not be an element of $\E^0$.
\end{proof}


The algebra $C$ is graded in a natural way by $G := \N[p^{-1}]^n$. The homogeneous
elements of $C$ of degree $\beta \in G$ are of the form $c\term{x}^{\beta}$ for
some $c \in K$. This grading extends to $\Omega_C$ by defining $\omega \in \Omega_C^r$
to be homogeneous of degree $\beta \in G$ if its coordinates in~\eqref{eq:generalForm}
are. We denote the homogeneous part of $\omega$ of degree $\beta$ by~$(\omega)_\beta$.

\begin{lem}[Explicit filtration {\cite[Proposition I.2.12]{bib:Il79}}] \label{lem:GradedComponentsFiltration}
	Let $\ell \ge 0$ and let $\beta \in G$. Define
	$\nu(\ell+1, \beta) := \min\bigl\{ \max\{ 0, \ell+1+v_p(\beta)\}, \ell+1 \bigr\} \in [0, \ell+1]$.
	Then $(\Fil^{\ell+1}\E)_\beta = p^{\nu(\ell+1, \beta)} (\E)_\beta$.
\end{lem}


The following lemma shows how degeneracy is naturally related to $\nu$. A proof is given in
\S\ref{app:ProofsWJCriterion}.

\begin{lem} \label{lem:DegeneracyCharacterization}
	Let $\ell \ge 0$ and let $f \in B\subset \E^0$. Then $f$ is $(\ell+1)$-degenerate if and only if
	the coefficient of $\term{x}^\beta$ in $\Frob^{-\ell} f$ is divisible by $p^{\nu(\ell+1, \beta)}$
	for all $\beta\in G$.
\end{lem}

\begin{lem}[Zeroness vs.~degeneracy] \label{lem:WJDegeneracyCharacterization}
	Let $\ell \ge 0$, let $\term{g}_r \in B \subset \E^0$ be polynomials, and define
	$\omega := d (\Frob^{-\ell}g_1)^{p^\ell} \wedge \dotsb \wedge d (\Frob^{-\ell}g_r)^{p^\ell} \in \E^r$.
	Then $\omega \in \Fil^{\ell+1}\E^r$ if and only if $\WJP_{\ell+1, I}(\term{g}_r)$ is $(\ell+1)$-degenerate 
	for all $I \in \tbinom{[n]}{r}$.
\end{lem}
\begin{proof}
	From the formula $d\Frob=p\Frob d$~\cite[(I.2.2.1)]{bib:Il79} we infer
        \[
        \Frob^\ell d (\Frob^{-\ell}g_i)^{p^\ell}=\Frob^\ell d \Frob^{-\ell}(g_i^{p^\ell})
        =p^{-\ell} d g_i^{p^\ell}= g_i^{p^\ell-1} dg_i,
        \]
        hence $\Frob^\ell \omega = (g_1 \dotsb g_r)^{p^\ell-1} \, d g_1 \wedge \dotsb \wedge d g_r$.
        A standard computation shows
	\[
		d g_1 \wedge \dotsb \wedge d g_r = \sideset{}{_I}\sum {\textstyle \bigl( \prod_{j \in I} x_j \bigr)} 
		\cdot \det \cJ_{\term{x}_I}(\term{g}_r) \cdot {\textstyle\bigwedge_{j \in I} d \log x_j},
	\]
	where the sum runs over all $I \in \tbinom{[n]}{r}$. This yields the unique representation
	\[
		\omega = \sideset{}{_I}\sum \Frob^{-\ell}\WJP_{\ell+1, I}(\term{g}_r) \cdot {\textstyle\bigwedge_{j \in I} d \log x_j}.
	\]
	By Lemma \ref{lem:GradedComponentsFiltration}, we have
	$\Fil^{\ell+1}\E^r = \bigoplus_{\beta \in G} \bigl(\Fil^{\ell+1}\E^r\bigr)_\beta
		= \bigoplus_{\beta \in G} p^{\nu(\ell+1, \beta)} \bigl(\E^r\bigr)_\beta$,
        and we conclude
        \begin{align*}
        \omega\in\Fil^{\ell+1}\E^r & \iff \forall\beta\in G\colon\ (\omega)_\beta\in p^{\nu(\ell+1, \beta)} \bigl(\E^r\bigr)_\beta \\
        & \iff \forall\beta\in G, I\in\tbinom{[n]}{r}\colon\ (\Frob^{-\ell}\WJP_{\ell+1, I}(\term{g}_r))_\beta\in p^{\nu(\ell+1, \beta)} \Frob^{-\ell}B \\
        & \iff \forall I\in\tbinom{[n]}{r}\colon\ \WJP_{\ell+1, I}(\term{g}_r)\text{ is }(\ell+1)\text{-degenerate},
        \end{align*}
        where we used Lemma \ref{lem:DegeneracyCharacterization}.
\end{proof}

\begin{proof}[Proof of Theorem \ref{thm:ExplicitWJCriterion}]
	Using Lemmas \ref{lem:TeichmuellerLiftAndFiltration} and \ref{lem:WJDegeneracyCharacterization},
	this follows from Theorem \ref{thm:AbstractWJCriterion}.
\end{proof}

\section{Independence testing: Proving Theorem \ref{thm:UpperBound}} \label{sec:UpperBound}

In this section, let $A = k[\term{x}]$ be a polynomial ring over an algebraic
extension~$k$ of~$\F_p$. For the computational problem of algebraic independence testing, we
consider $k$ as part of the input, so we may assume that $k = \F_{p^e}$ is a finite field.
The algorithm works with the truncated Witt ring $\W_{\ell+1}(\F_{p^t})$ of a small
extension $\F_{p^t}/k$. For computational purposes, we will use the fact that
$\W_{\ell+1}(\F_{p^t})$ is isomorphic to the \emph{Galois ring} $G_{\ell+1,t}$ of
characteristic $p^{\ell+1}$ and size $p^{(\ell+1)t}$ (see \cite[(3.5)]{bib:Rag69}).

This ring can be realized as follows. There exists a monic polynomial $h \in \Z/(p^{\ell+1})[x]$
of degree $t$ dividing $x^{p^t-1}-1$ in $\Z/(p^{\ell+1})[x]$, such that
$\ol{h} := h \pmod{p}$ is irreducible in $\F_p[x]$, and $\ol{\xi} := x + (\ol{h})$ is a primitive
$(p^t-1)$-th root of unity in $\F_p[x]/(\ol{h})$. Then we may identify
$G_{\ell+1,t} = \Z/(p^{\ell+1})[x]/(h)$ and $\F_{p^t} = \F_p[x]/(\ol{h})$, and
$\xi := x + (h)$ is a primitive $(p^t-1)$-th root of unity in $G_{\ell+1,t}$ 
(see the proof of \cite[Theorem 14.8]{bib:Wan03}).
The ring $G_{\ell+1,t}$ has a unique maximal ideal $(p)$ and $G_{\ell+1,t}/(p) \cong \F_{p^t}$.
Furthermore, $G_{\ell+1,t}$ is a free $\Z/(p^{\ell+1})$-module with basis
$1, \xi,\ldots,\xi^{t-1}$,
so that any $\ol{a} \in \F_{p^t}$ can be lifted coordinate-wise to $a \in G_{\ell+1,t}$
satisfying $\ol{a} \equiv a \pmod{p}$.
To map elements of $k$ to $\F_{p^t}$ efficiently, we use \cite{bib:Len91}. 

\comment{
The following lemmas show that the coefficient of a monomial in an arithmetic circuit over 
$G_{\ell+1,t}[z]$ can be computed by a $\sP$-oracle (see \cite{bib:Val79} for a definition of $\sP$).
The proofs are given in \S\ref{app:ProofsUpperBound}. 
}

For detailed proofs of the following two lemmas see \S\ref{app:ProofsUpperBound}.

\begin{lem}[Interpolation] \label{lem:Interpolation}
	Let $f \in G_{\ell+1,t}[z]$ be a polynomial of degree $D < p^t-1$ and let
	$\xi \in G_{\ell+1,t}$ be a primitive $(p^t-1)$-th root of unity. Then
	\[
		\coeff(z^d,f) = (p^t-1)^{-1} \cdot {\textstyle\sum_{j=0}^{p^t-2} \xi^{-jd} f(\xi^j)} \quad\text{for all }d \in [0,D].
	\]
\end{lem}

\noindent
This exponentially large sum can be evaluated using a $\sP$-oracle
 \cite{bib:Val79}.
\begin{lem}[$\sP$-oracle] \label{lem:SharpPOracle}
	Given $G_{\ell+1,t}$, a primitive $(p^t-1)$-th root of unity $\xi \in G_{\ell+1,t}$, an arithmetic circuit $C$
	over $G_{\ell+1,t}[z]$ of degree $D < p^t-1$ and $d \in [0,D]$. The $\coeff(z^d,C)$ can be computed in 
	$\FP^{\sP}$ (with a single $\sP$-oracle query).
\end{lem}

\begin{proof}[Proof of Theorem \ref{thm:UpperBound}]
	We set up some notation. Let $s := \size(\term{C}_r)$
	be the size of the input circuits. Then $\delta := 2^{s^2}$ is an upper bound for their degrees. 
	Set $\ell := \lfloor r \log_p \delta \rfloor$ and $D := r \delta^{r+1}+1$. 
	The constants of $\term{C}_r$ lie in $k = \F_{p^e}$, which is also given as input. 
	Let $t \ge 1$ be a multiple of $e$ satisfying $p^t - 1 \ge D^n$.
	Theorem \ref{thm:ExplicitWJCriterion} implies that the following procedure
decides the algebraic independence of $\term{C}_r$.
	\begin{enumerate}[(1)]
		\item\label{thm:UpperBoundA} Using non-determinism, guess $I \in \tbinom{[n]}{r}$ and $\alpha \in [0,D-1]^n$.
		\item\label{thm:UpperBoundB} Determine $G_{\ell+1,t}$ and $\xi$ as follows. Using non-determinism, guess
			a monic degree-$t$ polynomial $h \in \Z/(p^{\ell+1})[x]$. Check that $h$ divides $x^{p^t-1}-1$,
			$\ol{h} := h \pmod{p}$ is irreducible and $\ol{\xi} := x + (\ol{h})$ has order $p^t-1$ (for the last test,
			also guess a prime factorization of $p^t-1$), otherwise {\tt reject}.
			Set $\xi := x + (f)$.
		\item\label{thm:UpperBoundC} By lifting the constants of $\term{C}_r$ from $k$ to $G_{\ell+1,t}$, 
			compute circuits $\term{C}'_r$ over $G_{\ell+1,t}[\term{x}]$ such that 
			$\term{C}'_i \equiv \term{C}_i \pmod{p}$. Furthermore, compute a circuit $C$ for 
			$\WJP_{\ell+1,I}(\term{C}'_r)$ over $G_{\ell+1,t}[\term{x}]$.
		\item\label{thm:UpperBoundD} Compute the univariate circuit
			$C' := C(z, z^D, \dotsc, z^{D^{n-1}})$ over $G_{\ell+1,t}[z]$. The term $\term{x}^\alpha$
			is mapped to $z^d$, where $d := \sum_{i=1}^n \alpha_i D^{i-1}$.
		\item\label{thm:UpperBoundE} 
Compute $c := \coeff(z^d, C') \in G_{\ell+1,t}$. If $c$ is divisible
by $p^{\min\{v_p(\alpha),\ell\}+1}$, then {\tt reject}, otherwise {\tt accept}.
	\end{enumerate}
        In step	(\ref{thm:UpperBoundB}), the irreducibility of $\ol{h}$ can be
        tested efficiently by checking whether $\gcd(\ol{h}, x^{p^i}-x) = 1$ for 
	$i \le \lfloor t/2 \rfloor$ (see \cite[Theorem 10.1]{bib:Wan03}). 
        For the order test verify $\ol{\xi}^j \neq 1$ for all maximal divisors $j$
	of $p^t-1$ (using its prime factorization).

	The lifting in step (\ref{thm:UpperBoundC}) can be done as described in the beginning of the
	section. To obtain $C$ in polynomial time, we use \cite{bib:BS83} and \cite{bib:Ber84}
        for computing partial derivatives and the determinant,
	and repeated squaring for the high power.
	
	We have $\deg(C)\le r\delta(p^\ell-1)+r+r(\delta-1) \le r \delta^{r+1} < D$, so
	the \emph{Kronecker substitution} in step (\ref{thm:UpperBoundD}) preserves terms.
	Since $\deg_z(C') < D^n \le p^t-1$,
step (\ref{thm:UpperBoundE}) is in $\FP^{\sP}$ by Lemma \ref{lem:SharpPOracle}. Altogether we get an $\NP^{\sP}$-algorithm.
\end{proof} 



\section{Identity testing: Proving Theorem \ref{thm:HittingSet}} \label{sec:ApplicationPIT}

The aim of this section is to construct an efficiently computable hitting-set
for poly-degree circuits involving input polynomials of constant transcendence degree
and small sparsity, which works in any characteristic. It will involve sparse
PIT techniques and our Witt-Jacobian criterion. We use some lemmas from
\S\ref{app:ProofsApplicationPIT}.

As before, we consider a polynomial ring $A = k[\term{x}]$
over an algebraic extension~$k$ of $\F_p$.
Furthermore, we set $R:=\W(k)$ and $B:=R[\term{x}]$.
For a prime $q$ and an integer~$a$ we denote by $\lfloor a \rfloor_q$
the unique integer $0 \le b < q$ such that $a \equiv b \pmod{q}$.
Finally, for a polynomial $f$ we denote by $\sparse(f)$ its sparsity.


\begin{lem}[Variable reduction] \label{lem:FaithfulSubstitution}
	Let $\term{f}_r \in A$ be polynomials of sparsity at most $s \ge 1$ and degree at most $\delta \ge 1$.
	Assume that $\term{f}_r$, $\term{x}_{[r+1,n]}$ are algebraically independent.
	Let $D := r \delta^{r+1} + 1$ and let $S \subseteq k$ be of size
	$\abs{S} = n^2(2 \delta rs)^{4r^2s} \lceil \log_2 D \rceil^2 D$.
	
	Then there exist $c \in S$ and a prime 
	$2\le q\le n^2(2 \delta rs)^{4r^2s} \lceil \log_2 D \rceil^2$ such that
	$f_1(\term{x}_{r}, \boldsymbol{c}), \dotsc, f_r(\term{x}_{r}, \boldsymbol{c}) 
	\in k[\term{x}_{r}]$ are algebraically independent over $k$, where $\boldsymbol{c} = 
	\bigl( c^{\lfloor D^0 \rfloor_q}, c^{\lfloor D^1 \rfloor_q}, \dotsc, c^{\lfloor D^{n-r-1} \rfloor_q} \bigr)
	\in k^{n-r}$.
\end{lem}
\begin{proof}
	Let $g_i\in B$ be obtained from $f_i$ by lifting each coefficient, so
        that $g_i$ is $s$-sparse and $f_i \equiv g_i \pmod{pB}$.
	Theorem \ref{thm:ExplicitWJCriterion} implies that with $\ell := \lfloor r \log_p \delta \rfloor$
	the polynomial $g := \WJP_{\ell+1,[n]}(\term{g}_r, \term{x}_{[r+1,n]}) \in B$ is not
	$(\ell+1)$-degenerate.
	We have
	\begin{align*}
		g &= (g_1 \dotsb g_r \cdot x_{r+1} \dotsb x_n)^{p^\ell -1} (x_1 \dotsb x_n) \cdot 
		\det \cJ_{\term{x}}(\term{g}_r, \term{x}_{[r+1,n]}) \\
		&= (x_{r+1} \dotsb x_n)^{p^\ell} \cdot (g_1 \dotsb g_r)^{p^\ell -1}(x_1 \dotsb x_r) \cdot
			\det \cJ_{\term{x}_{r}}(\term{g}_r),
	\end{align*}
	since the Jacobian matrix $\cJ_{\term{x}}(\term{g}_r, \term{x}_{[r+1,n]})$ is
	block-triangular with the lower right block being the $(n-r) \times (n-r)$ identity matrix. 
	Define 
	\[
		g' := (g_1 \dotsb g_r)^{p^\ell -1}(x_1 \dotsb x_r) \cdot
		\det \cJ_{\term{x}_{r}}(\term{g}_r) \in B.
	\]
	Then $g = (x_{r+1} \dotsb x_n)^{p^\ell} g'$, and
        $g'$ is not $(\ell+1)$-degenerate by Lemma~\ref{lem:cancelDegen}. 
	Furthermore, we have $\deg(g') \le r\delta(p^\ell-1)+r+r(\delta-1) \le
	r \delta^{r+1} < D$ and
	\[
		\sparse(g') \le \binom{s + (p^\ell-1) - 1}{s-1}^r \cdot r! s^r 
		\le \bigl(s+\delta^{r}\bigr)^{rs}\cdot (rs)^r
		\le (2 \delta rs)^{2r^2s}.
	\]
	By Lemma \ref{lem:PreservingNonDegeneracy}, there exist $c \in S$ and a
        prime $q \le n^2(2 \delta rs)^{4r^2s} \lceil \log_2 D \rceil^2$ such that
	$h := g'(\term{x}_{r}, \boldsymbol{c}') \in R[\term{x}_{r}]$ is not
	$(\ell+1)$-degenerate, where 
        \[
        \boldsymbol{c} := \bigl( c^{\lfloor D^0 \rfloor_q}, c^{\lfloor D^1 \rfloor_q},
        \dotsc, c^{\lfloor D^{n-r-1} \rfloor_q} \bigr)\in k^{n-r},
        \]
        and $\boldsymbol{c}'\in R^{n-r}$ is the componentwise lift of $\boldsymbol{c}$ to $R$.
	Since
    $h = \WJP_{\ell+1,[r]}\bigl(g_1(\term{x}_{r}, \boldsymbol{c}')$, $\dotsc, g_r(\term{x}_{r}, \boldsymbol{c}')\bigr)$
	and $f_i(\term{x}_{r}, \boldsymbol{c}) \equiv g_i(\term{x}_{r}, \boldsymbol{c}') \pmod{pB}$ for all $i \in [r]$,
	Theorem \ref{thm:ExplicitWJCriterion} implies that
        $f_1(\term{x}_{r}, \boldsymbol{c})$, $\dotsc, f_r(\term{x}_{r}, \boldsymbol{c})$
        are algebraically independent over $k$.
\end{proof}

For an index set $I = \{i_1 < \dotsb < i_r\} \in \tbinom{[n]}{r}$ denote its complement by
$[n] \setminus I = \{ i_{r+1} < \dotsb < i_n \}$. Define the map
$\pi_I \colon k^n \rightarrow k^n$, $(a_1, \dotsc, a_n) \mapsto (a_{i_1}, \dotsc, a_{i_n})$.
We now restate, in more detail, and prove Theorem \ref{thm:HittingSet}.

\begin{thm}[Hitting-set]
	Let $\term{f}_m \in A$ be $s$-sparse, of degree at most
        $\delta$, having transcendence degree at most $r$, and assume $s,\delta,r\ge 1$.
	Let $C \in k[\term{y}_m]$ such that the degree of $C(\term{f}_m)$ is bounded by $d$.
	Define the subset
	\[
		\mathcal{H} := \Bigl\{ \pi_I\bigl( \boldsymbol{b}, 
		c^{\lfloor D^0 \rfloor_q}, c^{\lfloor D^1 \rfloor_q}, \dotsc, c^{\lfloor D^{n-r-1} \rfloor_q} \bigr) \,\bigl\vert\; 
		\text{$I \in \tbinom{[n]}{r}$, $\boldsymbol{b} \in S_1^r$, $c \in S_2$, $q \in [N]$} \Bigr\}
	\]
	of $k^n$, where $S_1, S_2 \subseteq k$ are arbitrary subsets of size $d+1$ and $n^2 (2 \delta rs)^{9 r^2s}$
	respectively, $D := r \delta^{r+1} + 1$, and $N := n^2 (2 \delta rs)^{7 r^2s}$.
	
	If $C(\term{f}_m)\ne0$ then there exists
	$\boldsymbol{a} \in \mathcal{H}$ such that
        $\bigl(C(\term{f}_m)\bigr)(\boldsymbol{a}) \neq 0$. The set
        $\mathcal{H}$ can be constructed in
        $\poly\bigl( (nd)^r, (\delta rs)^{r^2 s}\bigr)$-time.
\end{thm}
\begin{proof}
	We may assume that $\term{f}_r$ are algebraically independent over
        $k$
	There exists $I = \{ i_1 < \dotsb < i_r \}\subseteq[n]$ with complement
	$[n] \setminus I = \{ i_{r+1} < \dotsb < i_n \}$ such that $\term{f}_r$, $\term{x}_{[n] \setminus I}$
	are algebraically independent. By the definition of~$\mathcal{H}$, we may assume that
	$I = [r]$. By Lemma \ref{lem:FaithfulSubstitution}, there exist $c \in S_2$ and a prime $q \in [N]$
	such that $f_1(\term{x}_{r}, \boldsymbol{c}), \dotsc, f_r(\term{x}_{r}, \boldsymbol{c}) 
	\in k[\term{x}_{r}]$ are algebraically independent, where $\boldsymbol{c} = 
	\bigl( c^{\lfloor D^0 \rfloor_q}, c^{\lfloor D^1 \rfloor_q}, \dotsc, c^{\lfloor D^{n-r-1} \rfloor_q} \bigr)
	\in k^{n-r}$. 
	If $C(\term{f}_m) \neq 0$, then Lemma \ref{lem:FaithfulIsInjectiveOnSubalgebra} implies that 
	$\bigl(C(\term{f}_m)\bigr)(\term{x}_{r}, \boldsymbol{c}) \neq 0$.
	From Lemma \ref{lem:CombinatorialNullstellensatz} we obtain
	$\boldsymbol{b} \in S_1$ such that $\bigl(C(\term{f}_m)\bigr)(\boldsymbol{b}, \boldsymbol{c}) \neq 0$. Thus,
	$\boldsymbol{a} := (\boldsymbol{b}, \boldsymbol{c}) \in \mathcal{H}$
        satisfies the first assertion.
	The last one is clear by construction.
\end{proof}

\section{Discussion}

In this paper we generalized the Jacobian criterion for algebraic independence to any characteristic. 
The new criterion raises several questions. 
The most important one from the computational point of view:
Can the degeneracy condition in Theorem \ref{thm:ExplicitWJCriterion} be efficiently
tested? 
The hardness result for the general degeneracy problem shows that an affirmative
answer to that question must exploit the special structure of $\WJP$.
Anyhow, for constant or logarithmic $p$ an efficient algorithm for this problem
is conceivable.

In \S \ref{sec:ApplicationPIT}, we used the explicit Witt-Jacobian criterion to construct
faithful homomorphisms which are useful for testing polynomial identities. 
However, the complexity of this method is exponential in the sparsity of the given polynomials.
Can we exploit the special form of the $\WJP$ to improve the complexity bound?
Or, can we prove a criterion involving only the Jacobian polynomial (which in this case is sparse)?
(See an attempt in Theorem \ref{thm:Necessity}.)


\subsection*{Acknowledgements} We are grateful to the Hausdorff Center for Mathematics, Bonn, for its
kind support. J.M.~would like to thank the Bonn International Graduate School in 
Mathematics for research funding. N.S.~thanks Chandan Saha for explaining his results on finding coefficients
of monomials in a circuit \cite{bib:KS11}. We also thank Stefan Mengel for pointing
out the hardness of the degeneracy-problem.

\bibliographystyle{amsalpha}
\bibliography{refs}

\appendix

\section{Missing theorems, lemmas and proofs} \label{app:Proofs}

In this appendix we present statements and proofs that did not fit in the main part
due to space constraints.

\subsection{Degeneracy of the \texorpdfstring{$p$}{p}-adic Jacobian}

\begin{thm}[Necessity] \label{thm:Necessity}
	Let $\term{f}_r \in A$ and $\term{g}_r \in B$ such that $\forall i\in[r], f_i \equiv g_i \pmod{pB}$.
	If $\term{f}_r$  are algebraically dependent, then for any $r$ variables $\term{x}_I$, 
	$I \in \tbinom{[n]}{r}$,  the $p$-adic polynomial 
	$\hat{J}_{\term{x}_I}(\term{g}_r):=\bigl( {\textstyle\prod_{j \in I} x_j} \bigr) \cdot 
	\det\cJ_{\term{x}_I}(\term{g}_r)$ is degenerate.
	The converse does not hold.
\end{thm}
\begin{proof}
Fix $\ell\in\N$ such that $p^\ell$ is at least the degree of $\hat{J}_{\term{x}_I}(\term{g}_r)$.
Consider the differential form $\gamma:= 
d\V^{\ell}[f_1]_{\le\ell+1} \wedge \dotsb \wedge d\V^{\ell}[f_r]_{\le\ell+1} \in \W_{\ell+1}\Om^r_A$.

Assume that $f_1, \dotsc, f_r$ are algebraically dependent and set $R := k[f_1, \dotsc, f_r]$.
Corollary~\ref{cor:DiffModuleZero} implies $\W_{\ell+1}\Om^r_R = 0$, thus $\gamma$ vanishes in 
$\W_{\ell+1}\Om^r_R$. The inclusion $R \subseteq A$ induces a homomorphism $\W_{\ell+1}\Om_R^r \to \W_{\ell+1}\Om_A^r$, hence $\gamma$ vanishes in $\W_{\ell+1}\Om^r_A$ itself.

As in the proof of Lemma \ref{lem:TeichmuellerLiftAndFiltration}, we first make $\V^{\ell}[f]_{\le\ell+1}$
explicit. Let $g\in B$ such that $f\equiv g \pmod{pB}$, and write $g = \sum_{i=1}^s c_i \term{x}^{\alpha_i}$, 
where $c_i \in \W(k)$ and $\alpha_i \in \N^n$ for $i \in [s]$.
Note that $\Frob^\ell(\V^{\ell}[f]_{\le\ell+1})=p^\ell [f]_{\le\ell+1}$. Also, for 
$w:=\V^{\ell}(\sum_{i=1}^s c_i [\term{x}^{\alpha_i}] )\in \W_{\ell+1}(A)$ we have 
$\Frob^\ell(w)=p^\ell\sum_{i=1}^s c_i [\term{x}^{\alpha_i}]$. Since by assumption 
$([f]-\sum_{i=1}^s c_i [\term{x}^{\alpha_i}])\in \V\W(A)$, we get 
$p^{\ell}([f]-\sum_{i=1}^s c_i [\term{x}^{\alpha_i}])\in \V^{\ell+1}\W(A)$. This proves 
$\Frob^\ell(\V^{\ell}[f]_{\le\ell+1})=\Frob^\ell(w)$. The injectivity of $\Frob^\ell$ implies 
$\V^{\ell}[f]_{\le\ell+1}=w$. Finally, we can apply $\tau\colon \W_{\ell+1}(A) \rightarrow \E_{\ell+1}^0$
to get: $\tau(\V^{\ell}[f]_{\le\ell+1})=\tau(w)=\V^\ell(g)$. 

Thus, we have the explicit condition $\gamma':=\tau(\gamma) = d\V^\ell(g_1)\wedge\dotsb\wedge d\V^\ell(g_r) 
\in \Fil^{\ell+1}\E^r$. Now we continue to calculate $\gamma'$ much like in Lemma \ref{lem:WJDegeneracyCharacterization}.
The formula $d\Frob=p\Frob d$ (see~\cite[(I.2.2.1)]{bib:Il79}) implies $d=\Frob dp\Frob^{-1}=\Frob d\V$, hence
$d=\Frob^\ell d\V^\ell$. We infer $\Frob^\ell d (V^{\ell}g_i)=dg_i$, hence 
$\Frob^\ell \gamma' = d g_1 \wedge \dotsb \wedge d g_r$. Furthermore,
	\[
		d g_1 \wedge \dotsb \wedge d g_r = \sideset{}{_I}\sum {\textstyle \bigl( \prod_{j \in I} x_j \bigr)} 
		\cdot \det \cJ_{\term{x}_I}(\term{g}_r)\cdot {\textstyle\bigwedge_{j \in I} d \log x_j},
	\]
	where the sum runs over all $I \in \tbinom{[n]}{r}$. This yields
	\[
		\gamma' = \sideset{}{_I}\sum \Frob^{-\ell}\hat{J}_{\term{x}_I}(\term{g}_r) 
			\cdot {\textstyle\bigwedge_{j \in I} d \log x_j},
	\]
  and this representation is unique.
	
As in the proof of Lemma  \ref{lem:WJDegeneracyCharacterization} we conclude 
        \begin{align*}
        \gamma'\in\Fil^{\ell+1}\E^r & \iff \forall\beta\in G\colon\ (\gamma')_\beta\in p^{\nu(\ell+1, \beta)} \bigl(\E^r\bigr)_\beta \\
        & \iff \forall\beta\in G, I\in\tbinom{[n]}{r}\colon\ (\Frob^{-\ell}\hat{J}_{\term{x}_I}(\term{g}_r))_\beta\in p^{\nu(\ell+1, \beta)} \Frob^{-\ell} B \\
        & \iff \forall I\in\tbinom{[n]}{r}\colon\ \hat{J}_{\term{x}_I}(\term{g}_r)\text{ is }(\ell+1)\text{-degenerate},
        \end{align*}
where we used Lemma \ref{lem:DegeneracyCharacterization}. Since our $\ell$ is large enough, this
is finally equivalent to the degeneracy of $\hat{J}_{\term{x}_I}(\term{g}_r)$. This finishes the proof
of one direction.

The converse is false, because if we fix $f_1:=x_1^p$ and $f_2:=x_2^p$, then 
$\hat{J}_{\term{x}_2}(x_1^p,x_2^p)=p^2 x_1^p x_2^p$. This is clearly degenerate, but $f_1,f_2$ are
algebraically independent.
\end{proof}

\subsection{Proofs for Section \ref{sec:Preliminaries}}
\label{app:Preliminaries}

For a polynomial $f$ in some polynomial ring $k[\term{x}_n]$ and a vector $\term{w}\in\N^n$, the 
{\em weighted-degree} is defined as 
$$\max\bigl\{\sum_{i=1}^n w_ie_i\ |\ \term{e}\in\N^n, \coeff(\term{x}^{\term{e}},f)\ne0\bigr\}.$$
For the following proof we need to define a map $\mu_{\term{w}}:k[\term{x}]\rightarrow k[\term{x}]$
that extracts the {\em highest weighted-degree part}. 
I.e.~for $f\in k[\term{x}]$ of weighted-degree $\delta$, $\mu_{\term{w}}(f)$ is the sum of the weighted-degree-$\delta$ terms in $f$. E.g.~$\mu_{(1,3)}(2x_1^2+3x_2)=3x_2$.  
Note that $\mu_{\term{w}}(f)=0$ iff $f=0$.
\begin{repthm}{thm:PerronsTheorem}[{\textbf{restated}}]
Let $k$ be a field, $\term{f}_{n} \in k[\term{x}]$ be algebraically independent, and set 
$\delta_i := \deg(f_i)$ for $i \in [n]$. 
Then $[k(\term{x}_n):k(\term{f}_n)] \le \delta_1 \dotsb \delta_{n}$.
\end{repthm}
\begin{proof}
Define for each $i\in[n]$ the {\em homogenization} $g_i:=z^{\delta_i}\cdot f_i(\term{x}/z) \in k[z,\term{x}]$
of $f_i$ with respect to degree $\delta_i$.

Firstly, $z,\term{g}_n$ are algebraically independent over $k$. Otherwise,
there is an irreducible polynomial $H\in k[\term{y}_{[0,n]}]$ such
that $H(z,\term{g}_n)=0$. Evaluation at $z=1$ yields $H(1,\term{f}_n)=0$.
The algebraic independence of $\term{f}_n$ implies $H(1,\term{y}_n)=0$, hence
$(y_0-1)| H(\term{y}_{[0,n]})$ by the Gauss Lemma.
This contradicts the irreducibility of $H$.

Thus, $d':=[k(z,\term{x}_n):k(z,\term{g}_n)]$ is finite. We will now compare it
with $[k(\term{x}_n):k(\term{f}_n)]=:d$. Denote the vector spaces $k(z,\term{x}_n)$
over $k(z,\term{g}_n)$ by $\VV'$, and 
$k(\term{x}_n)$ over $k(\term{f}_n)$ by $\VV$. Each of these vector spaces admits a finite basis consisting of 
monomials in $\term{x}_n$ only.

Suppose $S=\{\term{x}^\alpha|\alpha\in I\}$, for some $I\subset\N^n$, is a basis
of $\VV'$. Assume 
that
$$\sum_{\alpha\in I} h_\alpha(\term{f}_n)\cdot \term{x}^\alpha = 0$$
with some $h_{\alpha}\in k[\term{y}_n]$.
By {\em homogenizing} each term in this equation
with respect to the same sufficiently large degree, 
we obtain $h'_{\alpha}\in k[\term{y}_{[0,n]}]$ such that 
$$\sum_{\alpha\in I} h'_\alpha(z,\term{g}_n)\cdot \term{x}^\alpha = 0.$$
Since the $\term{x}^\alpha$ are linearly independent over $k(z,\term{g}_n)$, we
conclude $h'_\alpha(z,\term{g}_n)=0$, hence $h_\alpha(\term{f}_n)=0$ for all $\alpha$.
Thus, $d'\le d$.

Suppose $S=\{\term{x}^\alpha|\alpha\in I\}$, for $I\subset\N^n$, is a basis of $\VV$. If they are linearly dependent in 
$\VV'$, then there exist $h_{\alpha}\in k[\term{y}_{[0,n]}]$ such that 
\begin{equation}\label{eq:PerronExt}
\sum_{\alpha\in I} h_\alpha(z,\term{g}_n)\cdot \term{x}^\alpha = 0
\end{equation}
is a nontrivial equation. Let $\term{1}:=(1,\ldots,1)\in\N^{n+1}$, $\term{w}:=(1,\term{\delta}_n)$ and $h'_{\alpha}:=\mu_{\term{w}}(h_{\alpha})\in k[\term{y}_{[0,n]}]$.
Applying $\mu_{\term{1}}$ on (\ref{eq:PerronExt}) we get for some nonempty $J\subseteq I$ a nontrivial equation:
$$\sum_{\alpha\in J} h'_\alpha(z,\term{g}_n)\cdot \term{x}^\alpha = 0.$$
Since $h'_\alpha(z,\term{g}_n)$ is homogeneous and nonzero, it cannot be divisible by $(z-1)$. Thus, 
$h'_\alpha(1,\term{f}_n)\ne0$ and we get a nontrivial equation in $\VV$:
$$\sum_{\alpha\in J} h'_\alpha(1,\term{f}_n)\cdot \term{x}^\alpha = 0.$$
This contradicts the choice of $I$. Hence, $d\le d'$.

Finally, $d=d'$ and from \cite[Corollary 1.8]{bib:K96} we know $d'\le \delta_1 \dotsb \delta_{n}$.
\end{proof}

Now we use the notation of \S \ref{sec:WittVectors}.

\begin{lem}[$p$-th powering] \label{lem:PthPowerAndFiltration}
	Let $A$ be an $\F_p$-algebra and let $a, b \in \W(A)$ such that $a-b \in \V\W(A)$. Then
	$a^{p^\ell} - b^{p^\ell} \in \V^{\ell+1}\W(A)$ for all $\ell \ge 0$.
\end{lem}
\begin{proof}
	We use induction on $\ell$, where the base case $\ell = 0$ holds by assumption.
	Now let $\ell \ge 1$. By induction hypothesis, there is $c \in \V^\ell\W(A)$ such that 
	$a^{p^{\ell-1}} = b^{p^{\ell-1}} + c$. Using $\V\Frob = p$ and $p^{-1}\tbinom{p}{i} \in \N$ 
	for $i \in [p-1]$, we conclude
	$a^{p^\ell} -  b^{p^\ell} = \bigl( b^{p^{\ell-1}} + c \bigr)^p -  b^{p^\ell}
		= c^p + \sum_{i=1}^{p-1} p^{-1}\tbinom{p}{i} \V\Frob\bigl( b^{p^{\ell-1}(p-i)}c^i \bigr) 
		\in \V^{\ell+1}\W(A)$.
\end{proof}

\begin{lem}[Multinomials {\cite[Theorem 32]{bib:Si80}}] \label{lem:DivisibilityMultinomialCoefficient}
	Let $\ell, s \ge 1$ and let $\alpha \in \N^s$
	such that $\abs{\alpha} = p^\ell$. Then $p^{\ell-v_p(\alpha)}$ divides the multinomial
	coefficient $\tbinom{p^\ell}{\alpha} := \tbinom{p^\ell}{\alpha_1, \dotsc, \alpha_s}$.
\end{lem}

Now we use the notation of \S \ref{sec:DeRhamWittComplex} and consider a function field 
$L:=k(\term{x}_n)$ over a perfect field $k$.

\begin{replem}{lem:kerFrob}[{\textbf{restated}}]
We have $\ker\big(\W_{\ell+i}\Om_L^r\stackrel{\Frob^i}{\longrightarrow}\W_{\ell}\Om_L^r\big)\subseteq\Fil^\ell\W_{\ell+i}\Om_L^r$.
\end{replem}
\begin{proof}
Let $\omega\in\W_{\ell+i}\Om_L^r$ with $\Frob^i \omega=0$. Applying $\V^i\colon\W_{\ell}\Om_L^r\to \W_{\ell+i}\Om_L^r$
and noting that $\V^i\Frob^i=p^i$, we conclude that $p^i\omega=0$. Proposition I.3.4
of~\cite{bib:Il79} implies $\omega\in \Fil^{\ell}\W_{\ell+i}\Om_L^r$.
\end{proof}

\subsection{Proofs for Section \ref{sec:SeparabilityClassicalJacobianCriterion}} 
\label{app:ProofsClassicalJacobianCriterion}

\begin{repthm}{thm:JacobianCriterion}[{\textbf{restated}}]
	Let $\term{f}_r \in k[\term{x}]$ be polynomials. Assume that $k(\term{x})$ is 
	a separable extension of $k(\term{f}_r)$.
	Then, $\term{f}_r$ are algebraically independent over $k$ if and only if
	$\J_{k[\term{x}]/k}(\term{f}_r) \neq 0$.
\end{repthm}
\begin{proof}
	Let $\term{f}_r$ be algebraically independent over $k$. Since $k(\term{x})$ is separable
	over $k(\term{f}_r)$, 
	we can extend our system to a separating transcendence basis $\term{f}_n$
	of $k(\term{x})$ over $k$. Since $k[\term{f}_n]$ is isomorphic to
	a polynomial ring, we have $\J_{k[\term{f}_n]/k}(\term{f}_n) \neq 0$.
        Lemmas \ref{lem:DeRhamLocalization} and \ref{lem:DeRhamFiniteSeparableExtension} imply
	$\J_{k[\term{x}]/k}(\term{f}_n) \neq 0$, thus
	$\J_{k[\term{x}]/k}(\term{f}_r) \neq 0$.
	
	Now let $\term{f}_r$ be algebraically dependent over $k$. The polynomials
	remain dependent over the algebraic closure $L := \ol{k}$, which is perfect.
Hence, $L(\term{f}_r)$ is separable over $L$, and
	\cite[Corollary 16.17 a]{bib:Eis95} implies
	$r > \trdeg_L(L(\term{f}_r)) 
		= \dim_{L(\term{f}_r)} \Om_{L(\term{f}_r)/L}^1$.
	Thus $df_1,\dotsc,df_r$ are linearly dependent, so
	$\J_{L(\term{f}_r)/L}(\term{f}_r) = 0$, implying
	$\J_{L[\term{f}_r]/L}(\term{f}_r)$ $= 0$ by Lemma~\ref{lem:DeRhamLocalization}.
	The inclusion $L[\term{f}_r] \subseteq L[\term{x}]$ induces an
	$L[\term{f}_r]$-module homomorphism 
	$\Om_{L[\term{f}_r]/L}^r \rightarrow \Om_{L[\term{x}]/L}^r$,
	hence $\J_{L[\term{x}]/L}(\term{f}_r) = 0$.
        Lemma~\ref{lem:DeRhamBaseChange} implies $\J_{k[\term{x}]/k}(\term{f}_r) = 0$.
\end{proof}

\begin{rem}
Note that without the separability hypothesis algebraic dependence of the
$\term{f}_r$ still implies $\J_{k[\term{x}]/k}(\term{f}_r)=0$.
\end{rem}

\begin{replem}{lem:SeparabilityInLargeChar}[{\textbf{restated}}]
	Let $\term{f}_m \in k[\term{x}]$ have transcendence degree $r$ and
	maximal degree $\delta$, and assume that $\ch(k) = 0$ or $\ch(k) > \delta^r$. Then the extension
	$k(\term{x})/k(\term{f}_m)$ is separable.
\end{replem}
\begin{proof}
	In the case $\ch(k) = 0$ there is nothing to prove, so let $\ch(k) = p > \delta^r$.
After renaming polynomials
	and variables, we may assume that $\term{f}_r$, $\term{x}_{[r+1,n]}$ are algebraically
	independent over $k$. We claim that $\term{x}_{[r+1,n]}$ is a separating transcendence basis
	of $k(\term{x})/k(\term{f}_m)$.
A transcendence degree argument shows that they form a transcendence basis.
Hence it suffices to show that~$x_i$ is separable over
$K := k(\term{f}_m, \term{x}_{[r+1,n]})$ for all $i \in [r]$. 
By Theorem \ref{thm:PerronsTheorem}, we have $[k(\term{x}):K]
\le [k(\term{x}):k(\term{f}_r, \term{x}_{[r+1,n]})] \le \delta^r < p$. Therefore, the degree of
the minimal polynomial of $x_i$ over $K$ is $< p$, thus $x_i$ is indeed separable for all $i \in [r]$.
\end{proof}

\subsection{Proofs for Section \ref{sec:WJCriterion}} \label{app:ProofsWJCriterion}

We use the notation of \S \ref{sec:WJPolynomial}.

\begin{replem}{lem:DegeneracyCharacterization}[{\textbf{restated}}]
	Let $\ell \ge 0$ and let $f \in B\subset \E^0$. Then $f$ is $(\ell+1)$-degenerate if and only if
	the coefficient of $\term{x}^\beta$ in $\Frob^{-\ell} f$ is divisible by $p^{\nu(\ell+1, \beta)}$
	for all $\beta\in G$.
\end{replem}
\begin{proof}
The map $\Frob^{-\ell}$ defines a bijection between the terms of $f$ and the
terms of $\Frob^{-\ell}f$ mapping $c\term{x}^\alpha\mapsto u\term{x}^\beta$ with
$u=\Frob^{-\ell}(c)$ and $\beta=p^{-\ell}\alpha$. Since $\alpha\in\N^n$, we have
$v_p(\beta)=v_p(p^{-\ell}\alpha)=v_p(\alpha)-\ell\ge-\ell$, thus
$\nu(\ell+1, \beta)=\min\{\ell+v_p(\beta),\ell\}+1=\min\{v_p(\alpha),\ell\}+1$,
which implies the claim.
\end{proof}

\subsection{Proofs for Section \ref{sec:UpperBound}} \label{app:ProofsUpperBound}

We use the notation of \S \ref{sec:UpperBound}.

\begin{replem}{lem:Interpolation}[{\textbf{restated}}]
	Let $f \in G_{\ell+1,t}[z]$ be a polynomial of degree $D < p^t-1$ and let
	$\xi \in G_{\ell+1,t}$ be a primitive $(p^t-1)$-th root of unity. Then
	\[
		\coeff(z^d,f) = (p^t-1)^{-1} \cdot {\textstyle\sum_{j=0}^{p^t-2} \xi^{-jd} f(\xi^j)} 
		\quad\text{for all $d \in [0,D]$}.
	\]
\end{replem}
\begin{proof}
	Set $m := p^t-1$. Note that $m$ is a unit in $G_{\ell+1,t}$, because $m \notin (p)$.
	It suffices to show that $\sum_{j=0}^{m-1} \xi^{-jd} \xi^{ij} = m \cdot \delta_{di}$
	for all $d,i \in [0,m-1]$. This is clear for $d = i$, so let $d \neq i$. Then
	$\sum_{j=0}^{m-1} \xi^{-jd} \xi^{ij} = \sum_{j=0}^{m-1} \xi^{j(i-d)} = 0$, because
	$\xi^{i-d}$ is an $m$-th root of unity $\neq 1$.
\end{proof}

\begin{replem}{lem:SharpPOracle}[{\textbf{restated}}]
	Given $G_{\ell+1,t}$, a primitive $(p^t-1)$-th root of unity $\xi \in G_{\ell+1,t}$, an arithmetic circuit $C$
	over $G_{\ell+1,t}[z]$ of degree $D < p^t-1$ and $d \in [0,D]$. The $\coeff(z^d,C)$ can be computed in 
	$\FP^{\sP}$ (with a single $\sP$-oracle query).
\end{replem}
\begin{proof}
	Set $m := p^t-1$. As in \S \ref{sec:UpperBound}, we assume that
	$G_{\ell+1,t} = \Z/(p^{\ell+1})[x]/(h)$, where $\deg(h) = t$, and $\xi = x + (h)$. By Lemma \ref{lem:Interpolation}, 
	we have to compute a sum $S := \sum_{i=0}^{m-1} a_i$ with $a_i \in G_{\ell+1,t}$. Each summand $a_i$ can be
	computed in polynomial time, because $C$ can be efficiently evaluated.
        Since the number of summands in $S$ is exponential, we need the help of a $\sP$-oracle to compute it.
	
	Each $a_i$ can be written as $a_i = \sum_{j=0}^{t-1} c_{i,j} \xi^j$ with
	$c_{i,j} \in \Z/(p^{\ell+1})$. Thus, we can represent $a_i$ by a tuple
	$c_i \in [0, p^{\ell+1}-1]^t$ of integers, and a representation of $S$
	can be obtained by computing the componentwise integer sum $\term{s} = \sum_{i=0}^{m-1} c_i$.
	Set $N := m \cdot p^{\ell+1}$. Then $\term{s}, c_i \in [0,N-1]^t$, so we can encode the tuples $\term{s}$ and
	$c_i$ into single integers via the bijection
	\[
		\iota\colon [0,N-1]^t \rightarrow [0,N^t-1], \quad 
		(n_0, \dotsc, n_{t-1}) \mapsto {\textstyle \sum_{j=0}^{t-1} n_j N^j}.
	\]
	This bijection and its inverse are efficiently computable. Moreover, $\iota$ is compatible with the sum under consideration, 
	i.e. $\iota(\term{s}) = \sum_{i=0}^{m-1} \iota(c_i)$, thus we reduced our problem to the summation
	of integers which are easy to compute.
	
	To show that $\iota(\term{s})$ can be computed in $\sP$, we have to design a non-deterministic
	polynomial-time Turing machine that, given input as above, has exactly $\iota(\term{s})$ accepting computation paths.
	This can be done as follows. First we branch over all integers $i \in [0,m-1]$. 
	In each branch $i$, we (deterministically) compute the integer $\iota(c_i)$ and branch again into
	exactly $\iota(c_i)$ computation paths that all accept. This implies that the machine has altogether
	$\sum_{i=0}^{m-1} \iota(c_i) = \iota(\term{s})$ accepting computation paths.
\end{proof}

We now state here the claims proved by Mengel \cite{bib:M12}. Define the problem of {\em $\ell$-Degen} as:
Given a univariate arithmetic circuit computing $C(x)\in\Q_p[x]$, test whether $C(x)$ is $\ell$-degenerate. 
Note that for $\ell=1$ this is the same as the identity test $C(x)\equiv0 \pmod{p}$, which can be done in
randomized polynomial time (or ZPP). The situation drastically changes when $\ell>1$.

\begin{thm}{\cite{bib:M12}}\label{thm:CPHard}
For $\ell>1$, $\ell$-Degen is $\CeqP$-hard under $\ZPP$-reductions.
\end{thm}
\begin{proof}[Proof sketch]
Denote by ZMC the problem: Given $m\in\N$ and a univariate arithmetic circuit computing $C(x)\in\Q[x]$, 
test whether $\coeff(x^m,C(x))=0$. By \cite{bib:fmm12} ZMC is C$_=$P-hard. The idea is to reduce ZMC to $2$-Degen.
Randomly pick a sufficiently large prime $p$. Consider the circuit $C'(x):=px^{p-m}\cdot C(x)$. It can be shown that
$C'(x)$ is $2$-degenerate iff $\coeff(x^m,C(x))=0$.  
\end{proof}


\begin{cor}{\cite{bib:M12}}\label{cor:PHCollapse}
Let $\ell>1$. If $\ell$-Degen is in $\PH$ then $\PH$ collapses. 
\end{cor}
\begin{proof}[Proof sketch]
Classically, we have $$\text{PH} \subseteq \text{NP}^{\sP} \subseteq \text{NP}^{\text{C}_=\text{P}}.$$
By the theorem it now follows that $$\text{PH} \subseteq \text{NP}^{\text{ZPP}^{\ell\text{-Degen}}}
\subseteq \text{NP}^{\text{NP}^{\ell\text{-Degen}}}.$$
Thus, if $\ell$-Degen $\in\Sigma_i$ then PH $\subseteq\Sigma_{i+2}$.
\end{proof}

\subsection{Proofs for Section \ref{sec:ApplicationPIT}} \label{app:ProofsApplicationPIT}

We use the notation of \S \ref{sec:ApplicationPIT}.

\begin{lem}[Using sparsity] \label{lem:PreservingNonDegeneracy}
	Let $\ell \ge 0$ and let $g \in B$ be an $s$-sparse polynomial
        of degree less than $D \ge 2$ which is not $(\ell+1)$-degenerate. Let
        $S \subset R$ be a subset such that $\abs{S\setminus pR} = (ns \lceil\log_2 D\rceil)^2 D$,
        and let $r\in[n]$.

	Then there exist $c \in S$ and a prime $q \le (ns \lceil\log_2 D\rceil)^2$
        such that $g(\term{x}_{r}, \boldsymbol{c}) \in R[\term{x}_{r}]$ is not
        $(\ell+1)$-degenerate, where
	$\boldsymbol{c} := \bigl( c^{\lfloor D^0 \rfloor_q}, c^{\lfloor D^1 \rfloor_q},
        \dotsc, c^{\lfloor D^{n-r-1} \rfloor_q} \bigr)\in R^{n-r}$.
\end{lem}

%

\begin{proof}
	Write $g = \sum_{\beta \in \N^r} g_{\beta} \term{x}_{r}^{\beta}$ with $g_{\beta} \in R[\term{x}_{[r+1,n]}]$. 
	Since $g$ is not $(\ell+1)$-degenerate, there exists $\alpha \in \N^n$
        such that the coefficient $c_\alpha \in R$ of $\term{x}^\alpha$ in $g$ is
        not divisible by $p^{\min\{v_p(\alpha),\ell\}+1}$.
	Write $\alpha=(\alpha',\alpha'') \in \N^r\times\N^{n-r}$. 
	Since $c_\alpha$ is the coefficient of $\term{x}_{[r+1,n]}^{\alpha''}$ in
        $g_{\alpha'}$, 
        this polynomial cannot be divisible by
        $p^{\min\{v_p(\alpha),\ell\}+1}$.
        Our aim is to
        find $\boldsymbol{c} \in R^{n-r}$ such that $g_{\alpha'}(\boldsymbol{c})$
        is not divisible by $p^{\min\{v_p(\alpha),\ell\}+1}$, since then it is
        neither by the possibly higher power $p^{\min\{v_p(\alpha''),\ell\}+1}$.
	In other words, if we write $g_{\alpha'} = p^e g'$, where $g'$ is not
        divisible by $p$, we have an instance of PIT over the field $R/pR\cong k$.
	
	We solve it using a \emph{Kronecker substitution}, so consider the
        univariate polynomial $h' := g'\bigl( t^{D^0}, t^{D^1}, \dotsc, t^{D^{n-r-1}} \bigr) \in R[t]$
        in the new variable $t$. Since  $\deg g'=\deg g_{\alpha'}\le \deg g<D$,
        the substitution preserves terms, so $h'\notin pR[t]$.
        Furthermore, $h'$ is $s$-sparse and of degree $<D^n$.
        For any $q\in\N$, let $h_q$ be the polynomial obtained from $h'$ by
        reducing every exponent modulo $q$.
        By \cite[Lemma 13]{bib:BHLV09}, there are $<ns \log_2 D$ many primes~$q$
        such that $h_q\in pR[t]$.
        Since the interval $[N^2]$ contains at least $N$ primes
	for $N \ge 2$ (this follows e.g.\ from~\cite[Corollary~1]{bib:RS62}),
	there is a prime $q \le (ns \lceil\log_2 D\rceil)^2$ with
	$h_q \notin pR[t]$. Since 
	$\deg(h_q) < qD \le (ns \lceil\log_2 D\rceil)^2 D = \abs{S\setminus pR}$,
	there exists $c \in S$ with $h_q(c) \notin pR$.
\end{proof}

\begin{lem}[$p$-adic triangle is isosceles]\label{lem:pAdicTriangleInequ}
	Let $\alpha, \beta \in \Q^s$. Then $v_p(\alpha+\beta) \ge \min\{v_p(\alpha), v_p(\beta)\}$,
	with equality if $v_p(\alpha) \neq v_p(\beta)$.
\end{lem}


\begin{proof}
	Let $i \in [s]$ such that $v_p(\alpha+\beta) = v_p(\alpha_i+\beta_i)$. Then
	$v_p(\alpha+\beta) = v_p(\alpha_i+\beta_i) \ge \min\{v_p(\alpha_i), v_p(\beta_i)\}
	\ge \min\{v_p(\alpha), v_p(\beta)\}$.
	
	Now assume $v_p(\alpha) \neq v_p(\beta)$, say $v_p(\alpha) < v_p(\beta)$.
	Let $i \in [s]$ such that $v_p(\alpha) = v_p(\alpha_i)$. Then
	$v_p(\alpha_i) < v_p(\beta_i)$, therefore we obtain
	$v_p(\alpha+\beta) \le v_p(\alpha_i+\beta_i) = \min\{v_p(\alpha_i), v_p(\beta_i)\}
	= v_p(\alpha_i) = v_p(\alpha) = \min\{v_p(\alpha), v_p(\beta)\} \le v_p(\alpha+\beta)$.
\end{proof}

\begin{lem}\label{lem:cancelDegen}
	Let $\ell \ge 0$, let $g \in B$ and let $\alpha \in \N^n$ with $v_p(\alpha) \ge \ell$.
	Then $g$ is $(\ell+1)$-degenerate if and only if $\term{x}^\alpha \cdot g$ is
	$(\ell+1)$-degenerate.
\end{lem}


\begin{proof}
	It suffices to show that $\min\{v_p(\beta), \ell\} = \min\{v_p(\alpha+\beta), \ell\}$
	for all $\beta \in \N^n$. But the assumption implies that
$\min\{v_p(\beta), \ell\} = \min\{v_p(\alpha), v_p(\beta), \ell\}$, which is
$\le \min\{v_p(\alpha+\beta), \ell\}$ by Lemma~\ref{lem:pAdicTriangleInequ}
with equality, if $v_p(\alpha) \neq v_p(\beta)$.
	If $v_p(\alpha) = v_p(\beta)$, then $\min\{v_p(\beta), \ell\} = \min\{v_p(\alpha), \ell\} = \ell \ge \min\{v_p(\alpha+\beta),\ell\}$.
\end{proof}

Let $\term{f}_m \in A$ be polynomials and let $\varphi\colon k[\term{x}] \rightarrow k[\term{x}_{r}]$ 
be a $k$-algebra homomorphism. We say that $\varphi$ is \emph{faithful to $\term{f}_m$} if
$\trdeg_k(\term{f}_m) = \trdeg_k(\varphi(\term{f}_m))$.

\begin{lem}[Faithful is useful {\cite[Theorem 11]{bib:BMS11}}] \label{lem:FaithfulIsInjectiveOnSubalgebra}
	Let $\varphi\colon A \rightarrow k[\term{x}_{r}]$ 
	be a $k$-algebra homomorphism and $\term{f}_m \in A$. 
	Then, $\varphi$ is faithful to $\term{f}_m$ iff $\varphi\vert_{k[\term{f}_m]}$ is injective.
\end{lem}

\begin{lem}{\cite[Corollary 1]{bib:Sch80}} \label{lem:CombinatorialNullstellensatz}
	Let nonzero $f \in k[\term{x}_{r}]$,
	and $S \subseteq k$ with $|S|>\deg f$. Then there exists $\boldsymbol{b} \in S^r$ 
	such that $f(\boldsymbol{b}) \neq 0$.
\end{lem}

%
%
%

\end{document}